\documentclass[aps,pra,superscriptaddress,onecolumn,twoside,amsmath,amssymb,notitlepage,preprintnumbers]{revtex4-1}

\usepackage{graphicx}
\usepackage{dcolumn}
\usepackage{bm}
\usepackage{hyperref}
\usepackage{subfigure}
\hypersetup{
colorlinks=true,
linkcolor=blue,
filecolor=blue,
citecolor=blue,  
urlcolor=black,
}

\usepackage{comment}

\usepackage{amsthm}
\theoremstyle{plain}
\newtheorem{thm}{Theorem}
\newtheorem{lem}[thm]{Lemma}
\newtheorem{prop}[thm]{Proposition}

\newtheorem{defn}[thm]{Definition}

\newtheorem{exmp}[thm]{Example}

\newtheorem*{rem}{Remark}

\newcommand{\ket}[1]{|#1\rangle}
\newcommand{\bra}[1]{\langle #1|}

\newcommand{\ketbra}[2]{|#1\rangle\!\langle #2|}
\newcommand{\proj}[1]{|#1\rangle\!\langle #1|}

\newcommand{\id}{\operatorname{id}}
\newcommand{\1}{\openone}
\newcommand{\cN}{\mathcal{N}}
\newcommand{\cL}{\mathcal{L}}
\newcommand{\cD}{\mathcal{D}}
\newcommand{\cE}{\mathcal{E}}
\newcommand{\cM}{\mathcal{M}}

\newcommand{\cU}{\mathcal{U}}
\newcommand{\cV}{\mathcal{V}}
\newcommand{\cF}{\mathcal{F}}
\newcommand{\cS}{\mathcal{S}}
\newcommand{\CC}{\mathbb{C}}
\newcommand{\RR}{\mathbb{R}}

\newcommand{\ox}{\otimes}

\newcommand{\tr}{\operatorname{Tr}}

\begin{document}


\title{Resource theories of quantum channels \protect\\ and the universal role of resource erasure}

\author{Zi-Wen Liu}
\email{zwliu@mit.edu, zliu1@perimeterinstitute.ca}
\affiliation{Perimeter Institute for Theoretical Physics, Waterloo, Ontario  N2L 2Y5, Canada}
\affiliation{Center for Theoretical Physics, Massachusetts Institute of Technology, Cambridge, MA 02139, USA}
\affiliation{Research Laboratory of Electronics, Massachusetts Institute of Technology, Cambridge, MA 02139, USA}
\affiliation{Department of Physics, Massachusetts Institute of Technology, Cambridge, MA 02139, USA}

\author{Andreas Winter}
\email{andreas.winter@uab.cat}
\affiliation{ICREA---Instituci\'o Catalana de Recerca i Estudis Avan\c{c}ats, Pg.~Lluis Companys, 23, ES-08001 Barcelona, Spain} 
\affiliation{F\'{\i}sica Te\`{o}rica: Informaci\'{o} i Fen\`{o}mens Qu\`{a}ntics, 
Departament de F\'{\i}sica, Universitat Aut\`{o}noma de Barcelona, ES-08193 Bellaterra (Barcelona), Spain}

\begin{abstract}
We initiate the systematic study of resource theories of quantum channels,
i.e.~of the dynamics that quantum systems undergo by completely positive maps,
\emph{in abstracto}: Resources are in principle all maps from one
quantum system to another, but some maps are deemed \emph{free}.
The free maps are supposed to satisfy certain axioms, among them 
closure under tensor products, under composition and freeness of the 
identity map (the latter two say that the free maps form a monoid). 
The free maps act on the resources simply by tensor product and composition.
This generalizes the much-studied resource theories of quantum states,
and abolishes the distinction between resources (states) and the free
maps, which act on the former, leaving only maps, divided into resource-full 
and resource-free ones.

We discuss the axiomatic framework of quantifying channel resources, and show 
two general methods of constructing resource monotones of channels. 
Furthermore, we show that under mild regularity
conditions, each resource theory of quantum channels has a distinguished
monotone, the \emph{robustness} (and its smoothed version), generalizing
the analogous concept in resource theories of states.
We give an operational interpretation of the log-robustness as
the amount of heat dissipation (randomness) required for resource erasure by
random reversible free maps, valid in broad classes of resource
theories of quantum channels. Technically, this is based on an
abstract version of the recent \emph{convex-split lemma}, extended
from its original domain of quantum states to ordered vector
spaces with sufficiently well-behaved base norms (which includes the
case of quantum channels with diamond norm or variants thereof). 
Finally, we remark on several key issues concerning the asymptotic theory.\footnote{This work originated from a visit of ZWL to Barcelona in 2017. Preliminary results of it have been reported at various venues, including the ``Hong Kong-Shenzhen Workshop on Quantum Information Science'' in Shenzhen (China), May 2018, the ``Beyond IID in Information Theory'' workshop in Cambridge (UK), July 2018, and the ''Mathematics of Quantum Information'' conference in Siegen (Germany), March 2019.}
\end{abstract}

\date{8 April 2019}

\maketitle

\section{Resource theories: from states to maps}
\label{sec:intro}
The paradigm of resource theories has been applied successfully to
capture the essence of what is valuable, and how to measure its value,
in scenarios where certain objects and transformations are considered
relatively ``easy'' compared to others. In resource theories this
is idealized by considering some objects and transformations \emph{free},
meaning they may be invoked unlimitedly in any situation. Specifically,
resource theories of quantum states, i.e.~where the objects are states
of quantum systems undergoing (free) quantum channels, have proved to
be extremely successful in characterizing various quantum and other 
features of quantum states, such as 
entanglement \cite{virmani,PhysRevA.57.1619,entrev}, 
coherence \cite{PhysRevLett.113.140401,PhysRevLett.116.120404,RevModPhys.89.041003},
thermal non-equilibrium \cite{PhysRevLett.111.250404,nanothermo},
asymmetry (related to conserved quantities by Noether's theorem) 
\cite{PhysRevA.80.012307,noether}, magic states \cite{VMGE:magic,howard_2017}, etc. 
There are also several general theories that unify the common 
features of resource theories of states \cite{HorodeckiOppenheim:resource-theories,bg,rdm,PhysRevA.95.062314,regula,AnshuHsiehJain:erasure}.

It has been realized that the structure of resource theory can be formalized 
in much more abstract ways, and thus applied to more general settings. 
For instance, Refs.~\cite{COECKE201659,fritz_2017} establish algebraic frameworks
for resource theories, where the resources can essentially be any mathematical
object that has a combinable structure. 
In another direction, del Rio \emph{et al.} \cite{2015arXiv151108818D,delRio:currencies,2018arXiv180604937S} 
have attempted to formalize resource theories of knowledge, where rather than combining
systems, a top-down approach captures only subsystems of a global entity.

Here, rather than going all the way to these abstract structures,
we will explore a more modest, but for quantum mechanics highly
important and distinguished, extension, namely from quantum states to quantum channels
as the objects of the resource theory. 
There are a number of strong motivations for doing so, including but not restricted to the following few.  
On the one hand, quantum channels or processes can represent dynamical resources which, 
as opposed to static state resources, play natural roles in many physical scenarios. 
For example, certain quantum channels can be used to efficiently transmit quantum 
information, and certain thermodynamic processes can be implemented to do work.
Therefore, the resource theory approach for quantum channels is of great practical interest. 
On the other hand, due to the more complicated mathematical structure of quantum channels, 
the associated resource theory framework can be highly nontrivial and interesting from a 
mathematical point of view.  In particular, as will be discussed in more depth later, 
several key aspects of the resource theory approach such as resource composition and 
transformation become more subtle than the traditional state theory.
Furthermore, the application of resource theory approach to quantum channels augments 
and advances the study of this core area of quantum information, and of course 
the understanding of this extended resource theory scheme could greatly benefit 
from the profound literature of e.g.~quantum channel coding and capacities. 
Our programme is not entirely new, in fact it has been done, at least in part, 
for certain concrete quantum resource theories:
for instance for a good part of quantum Shannon theory \cite{DHW:Shannon-resource,qrst},
bipartite entanglement \cite{bip}, 
athermality \cite{thermal_c,FaistRenner:thermo-cost,FaistBertaTomamichel:thermo-cost},
and recently for the resource theories of coherence \cite{coh_c,coh_mio} and magic \cite{WangWildeSu:channel_magic}.
These prior works will serve us, among others, as instructive
concrete examples to look at, where the general structures may
seem abstract and uninformative.  
We would also like to mention that an ongoing work by Gour also contributes to this programme, 
and several recent developments on the entropies and relative entropies of 
channels \cite{GourWilde,Gour:superchannel,Yuan:channel-hypothesis} 
are relevant.

In this work, we will present a general framework of resource theories of
quantum channels, characterized by a set of free operations 
that have to satisfy certain axioms (Section~\ref{sec:free}). 
These are used to transform the resourceful channels, essentially
by composition and tensor product with free ones, following the structures of cptp quantum circuits  \cite{Aharonov:cptpcircuit} and quantum combs \cite{comb,Pavia-combs} (Section~\ref{sec:transformations}).
Transformations can be considered either exact, or probabilistic,
or approximate (the latter by default with respect to the diamond norm, 
the natural statistical distance on quantum channels, or some relative
of it, which appear naturally when considering multiple resources).
Unlike quantum states, where composition of systems and states is
captured perfectly by the tensor product, this topic is more
subtle for quantum channel, as we will discuss in Section~\ref{sec:composition}.
This is an aspect which has not attained the full attention it
deserves, and we hope to clarify it here in the general setting.

We will then show two important ways of constructing monotones
for these resource theories (Section~\ref{sec:measures}). We
also introduce a very important distinguished measure, the 
\emph{resource robustness} of a channel, which then we show to
have an operational interpretation as the heat dissipation 
required in the erasure of resourceness (Section~\ref{sec:erasure}). 
The erasure protocol is based on a generalized version 
of the recent so-called ``convex-split lemma'' for quantum states,
which we prove in the abstract setting of ordered vector
spaces equipped with a sufficiently well-behaved norm
(Appendix~\ref{app:convex-split}).

\medskip
Before we get started on the actual formalism of our general form
of resource theories of quantum channels, we comment briefly on the
setting. We will have to talk about different quantum systems, denoted
$A$, $B$, etc, each coming with its own complex Hilbert space, which we
conveniently denote by the same letter, not to overload notation.
We stress, however, that reference to ``a quantum system $A$'' is
much more than the Hilbert space; it comprises all physical features 
relevant to describe the system, be that a Hamiltonian and other
distinguished observables, a semigroup, a computational basis, a
tensor product structure (so as to have a meaningful entanglement
structure), etc. Not every quantum system can be treated within every
resource theory, it is the resource theory that has to identify ``its''
systems, and for each pair of systems $A$ and $B$, a subset
of ``free'' quantum channels (completely positive and trace preserving
maps, cptp). Formally, the mathematical structure naturally adapted to
describe this is a category, the objects being a certain class
of quantum systems, and the morphisms the cptp maps between them;
the resource theory of channels we are aiming at could be described
as a sub-category, having the same objects, but for each pair
of quantum systems a subset of ``free'' channels. We will not 
enforce this language, not to encumber or presentation with 
additional formal baggage. Instead, we will silently assume that
every system we speak about has sense in the resource theory at hand.
What we will however require is that there is a natural tensor product for
the admissible quantum systems, which in particular is represented
by the Hilbert space tensor product for the underlying state spaces,
and which is associative and commutative (up to natural isomorphism,
in the parlance of category theory).

Prominent examples of such theories include the following, to which we
will return at suitable points of our development:
\begin{itemize}
\item Quantum Shannon theory: resources are channels from Alice to Bob, 
      while all local channels are free \cite{DHW:Shannon-resource}.
\item Bipartite unitaries as resources, with LOCC as free operations \cite{bip}.
\item Thermodynamics of systems with non-interacting Hamiltonians: Gibbs-state-preserving maps
      are free. Work cost and work capacity of states and channels are of particular interest.
\item Coherence theory (with any of several classes of free operations, IO, MIO, DIO, etc).
      Again, of special interest are coherence generating capacity and coherence 
      cost \cite{coh_c,coh_mio}.
\item Entanglement-assisted communication: This is a resource theory where every state is free, yet
      it is nontrivial. The setting is of bipartite systems (Alice and Bob) with free entanglement and
      free local operations, and resources are, among others, channels from Alice to Bob; more
      generally, bidirectional channels between them. A variant has all no-signalling 
      channels free.
\end{itemize}

\section{Free resources}
\label{sec:free}
The free resources need to satisfy a set of axioms in order for the 
resource theory to exhibit a well-behaved structure.  
The conditions introduced in the following underlie a reasonable framework 
for a resource theory of quantum channels, and allow us to prove certain results.  
We shall first write down each mathematical statement, and then explain what
it means physically or operationally. 

Let $\mathfrak{F}$ be the class of free channels, which in a certain
precise sense defines the entire resource theory, as we shall see.
What we mean by this is that for each pair $A$ and $B$ of quantum systems,
for which we denote their Hilbert spaces by $A$ and $B$, respectively,
as well, we have a set $\mathfrak{F}(A\longrightarrow B)$
of channels, i.e.~cptp maps, so that 
$\mathfrak{F}(A\longrightarrow B) \subset \text{CPTP}(A\longrightarrow B)$.
We start with three necessary conditions; they should always hold, and 
will be repeatedly used in proofs.

\begin{enumerate}
  \item $\mathfrak{F}$ is closed under composition and tensor product. 
  \smallskip\\ \emph{Composition and tensor product represent the two possible ways to 
               combine two channels: acting them sequentially or in parallel. 
               The combination of two free channels should remain free.}
  \item All sets $\mathfrak{F}(A\longrightarrow B)$ are topologically closed, 
        i.e.~closed under limits.
  \smallskip\\ \emph{If a convergent sequence of channel is free, this means that the
               limiting channel can be approximated to arbitrary precision,
               and hence we would like to consider it free as well.
               Note that this requires a topology, which in the case of finite
               dimensional $A$ and $B$ is unique, but in general we shall
               silently assume the diamond norm topology.}
  \item For all $A$, $\mathfrak{F}(A\longrightarrow A)$ contains the identity
        channel $\id_A$.
  \smallskip\\ \emph{Doing nothing is always free. Together with item 1 this also
               guarantees that tensoring the identity is a free extension.}
\end{enumerate}

Following Gour, because of properties 1 and 3 above, we call free operations 
that satisfy the above axioms ``completely free''.

If one of the two Hilbert spaces $A$ or $B$ is trivial, i.e.~one-dimensional,
$\text{CPTP}(A\longrightarrow B)$ is special. On the one hand,
$\text{CPTP}(A\longrightarrow \CC)$ consists only of the trace $\tr_A$,
which maps every state on $A$ to $1$. On the other hand,
$\text{CPTP}(\CC\longrightarrow B)$ can be identified with the set
of states on $B$, denoted $\cS(B)$, because $\CC$ has only one state,
$1$, and each cptp map is entirely characterized by its image on that
point.

The following conditions do not seem necessary for all reasonable theories, 
but they are often desirable and can play key roles in some arguments. 
We list them as optional conditions:

\begin{enumerate}
  \setcounter{enumi}{3}
  \item For every system $A$, the partial trace over $A$ is free, 
        i.e.~$\tr_A \in\mathfrak{F}(A\longrightarrow\CC)$.
        \smallskip\\ 
        \emph{Ignoring subsystems is allowed.}
  \item Every system $B$ has some free states,
        i.e.~$\mathfrak{F}(\CC\longrightarrow B) \neq \emptyset$ is non-empty.
        \smallskip\\ 
        \emph{This is in some sense dual to the previous condition, and
              allows us to create a system in a free state.}
  \item All sets $\mathfrak{F}(A\longrightarrow B)$ of free channels are convex.
        \smallskip\\ 
        \emph{Probabilistic mixing does not create resource channels from free ones.}
  \item In systems composed of identical parts, the permutations are
        free operations, i.e.~for $A^n=A^{\otimes n}$, the cptp maps
        $\cU_\pi = U_\pi\cdot U_\pi^\dagger \in \mathfrak{F}(A^n\longrightarrow A^n)$,
        where $U_\pi$ is the unitary acting on the $n$ systems by
        permutation $\pi\in S_n$.
        \smallskip\\ 
        \emph{This is true in many cases, since permutations are nothing but
              a relabelling of the $n$ copies of $A$, which is considered
              purely conventional.}
\end{enumerate}

\medskip
Where does the set of free channels come from? In the first instance, 
as we have mentioned before, it may be the very physical expression of
what the resource theory is about, for instance local operations and
classical communication (LOCC) in the case of entanglement; or local
encodings and decodings in the resource theory formulation of Shannon
theory; or thermal operations in the resource theory of quantum
thermodynamics; or incoherent operations in the resource theory of
coherence. It can be checked that all of the named examples satisfy
properties 1 through 3, and usually all or most of 4 through 7.

However, in theories of state resources, where the free states tend to
play a special role, one can define a maximal set of free operations,
which consists of all operations that do not create resource from free states, 
i.e.~that map free states to free states, since allowing any other operation 
would trivialize the theory. 
For very broad classes of resource theories of states, it leads to a key 
feature of resource interconversions, namely \emph{asymptotic reversibility}:
The rate of interconverting asymptotically many copies of any state into 
any other state, is governed by the (regularized) relative entropy of 
resourceness \cite{bg}. It should be noted however that this maximal
set of operations tend to violate property 1 above. That is, while 
these sets or certainly closed under composition (concatenation)
of free maps, in many concrete examples, the tensor product of
free maps is no longer free; this happens for instance in the case
of entanglement, where the free states are the separable states and
the mentioned maximal set is the set of all cptp maps taking separable
to separable states \cite{BrandaoPlenio:ent-reversible}.
Consider on the other hand the simple example of the resource theory of coherence: 
It doesn't matter what initial set of incoherent operations one considers
for free, they all have in common the free incoherent states $\mathcal{I}$
\cite{PhysRevLett.113.140401}; this then defines the maximal incoherent operations (MIO)
\cite{2006quant.ph.12146A,RevModPhys.89.041003} as the largest set of
operations mapping incoherent to incoherent states. In this case we are 
lucky, and the theory satisfies all the axioms 1 through 7 above.
In cases like this it is interesting to ask whether the resulting
resource theory of channels is asymptotically reversible or not.
Before we can face this question, we have to consider how the free
channels act on the other channels as resource transformations.

\section{Resource transformations: the channel simulation framework}
\label{sec:transformations}
The fundamental question in a resource theory, and in fact the reason
for its existence, is which resource can be transformed to another one
by free operations. In the case of resource theories of quantum states
this is straightforward: the free operations act directly on the states.
Here, since the resources are also channels, it is both simpler and
more subtle: First, it is simpler because we can compose the given 
resource with a free channel, either as pre- or post-processing.
To attach ancillary systems we declare that adding a free transformation
to a given resource is always possible for free, so that in particular
the tensor product with a free resource is an allowed transformation.
Secondly, it is more subtle than in the case of states, because attaching ancillary 
systems cannot be treated as a cptp map itself. Conversely, removing a system
is not trivial, either, because of the asymmetry between input and
output systems of a channel: clearly an output system can be 
composed with the trace, the unique cptp map with range the complex 
numbers $\CC$, but to remove an input system, a state is required,
which is not unique.

However, before we build up the framework of resource transformation from
the bottom, it will be good to take a more abstract perspective, namely
that of ``supermaps'' \cite{Pavia-supermap}, or rather ``superchannels'', as we will call 
them following \cite{ChitambarGour:QRT,Gour:channel_resource}. These are
defined as linear transformations on the channels, themselves embedded in
the linear space of Hermitian-preserving maps, 
which have to preserve complete positivity and trace-preservation,
not only acting themselves but also tensored with the identity:
I.e., a superchannel taking maps $A\longrightarrow B$ to
maps $A'\longrightarrow B'$ should satisfy
\[
   \Theta\ox\id_R\bigl(\text{CPTP}(AR\longrightarrow BR)\bigr) \subset \text{CPTP}(A'R\longrightarrow B'R).
\]
In \cite{Pavia-supermap} it is shown that these requirements are equivalent to 
the following structure of $\Theta$ (as illustrated by Fig.~\ref{fig:sim}):
\begin{equation}
  \label{eq:structure-comb}
  \Theta(\cN) = \cD \circ (\cN\otimes\id_C) \circ \cE,
\end{equation}
with a suitable auxiliary system $C$ and ``encoding'' and
``decoding'' cptp maps $\cE:A' \longrightarrow A\otimes C$
and $\cD:B\otimes C \longrightarrow B'$, respectively.   
\begin{figure}[htbp]
    \centering
    \includegraphics[scale=0.45]{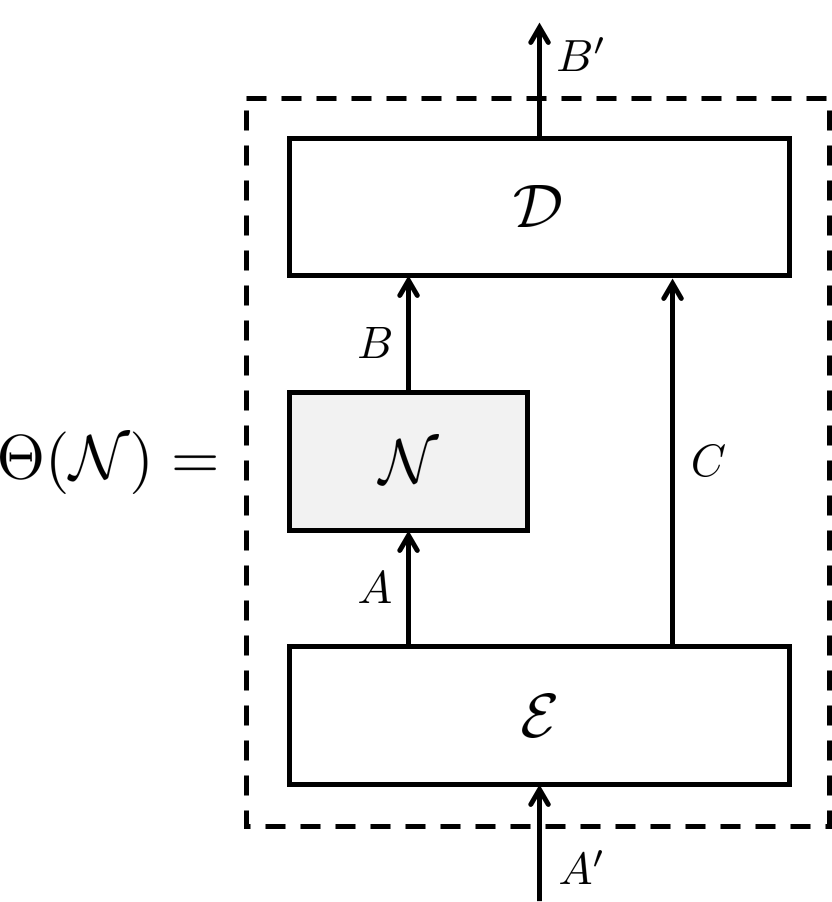}
    \caption{General scheme of channel simulation.}
    \label{fig:sim}
\end{figure}

What can we learn from this picture about resource theories of channels? 
On the one hand, $\Theta$ seems to be what we are looking for, a 
transformation of channels. For it to be admissible in the usual
framework of resource theories, it should map free channels to free channels;
in fact it should do so in a complete way, i.e.~even when tensored
with the identity on an auxiliary system $R$:
\begin{equation}\begin{split}
  \label{eq:freeness-preserving}
  \Theta\bigl(\mathfrak{F}(A\longrightarrow B)\bigr) &\subset \mathfrak{F}(A'\longrightarrow B'), \\
  \Theta\ox\id_R\bigl(\mathfrak{F}(AR\longrightarrow BR)\bigr) &\subset \mathfrak{F}(A'R\longrightarrow B'R).
\end{split}\end{equation}
Note that Gour's work \cite{Gour:superchannel} presents necessary and sufficient conditions for 
such channel conversion for a set of initial and final channels by means of a single superchannels, 
expressed in terms of the extended conditional min-entropy for quantum channels.
This is clearly a minimal requirement, as allowing any other superchannel would lead to the 
ability of creating resourceful channels from free ones, but it is less clear
whether we should allow all such superchannels. 
In particular, it would be important that the transformation takes the form of cptp quantum 
circuits \cite{Aharonov:cptpcircuit} or ``free'' quantum combs \cite{comb,Pavia-combs}.
Looking at the form of $\Theta$ in Eq.~(\ref{eq:structure-comb}), while this decomposition
is not unique at all, it is clear that $\Theta$ preserves the free channels
if $\cE$ and $\cD$ are both free, i.e.~$\cE\in\mathfrak{F}(A'\longrightarrow A\otimes C)$
and $\cD\in\mathfrak{F}(B\otimes C\longrightarrow B')$.
Here we will take the view that an admissible resource transformation
should have just such a decomposition. While this is mathematically
more complex than the freeness-perservation expressed in 
Eq.~(\ref{eq:freeness-preserving}), since to check eligibility we
have to search over all decompositions according to Eq.~(\ref{eq:structure-comb}),
we hold that it is conceptually much clearer, since now the resource
theory rises entirely out of $\mathfrak{F}$, and especially all
resource transformations are built up by constructing quantum
circuits from free channels.

\begin{defn}
\label{def:F-simulation}
We say that a channel $\cN:A\longrightarrow B$ can simulate a channel 
$\cN':A'\longrightarrow B'$ using the free resources $\mathfrak{F}$
(or equivalently that $\cN'$ can be implemented using $\cN$), if
there exists a system $C$, and $\cE\in\mathfrak{F}(A'\longrightarrow A\otimes C)$
and $\cE\in\mathfrak{F}(B\otimes C\longrightarrow B')$,
such that $\cN' = \cD \circ (\cN\otimes\id_C) \circ \cE$;
we shall use the notation $\cN \stackrel{\mathfrak{F}}{\longrightarrow} \cN'$.

If there exists a sequence of triples $(C,\cE,\cD)$
such that the channels $\cD \circ (\cN\otimes\id_C) \circ \cE$ converge
to $\cN'$, we speak of an asymptotic simulation or implementation,
using however the same notation as before unless it could result in
a confusion.

If $\frac12 \|\cD \circ (\cN\otimes\id_C) \circ \cE -\cN' \|_\diamond \leq \epsilon$,
we speak of an $\epsilon$-simulation or -implementation, and denote it 
$\cN \stackrel{\mathfrak{F}}{\longrightarrow} \; \stackrel{\epsilon}{\approx}\!\cN'$.
\end{defn}

This defines the free transformations from channels to channels;
by construction, they fulfill Eq.~(\ref{eq:freeness-preserving}).
However, in general there is no reason why conversely a superchannel $\Theta$
satisfying the latter condition has to have an implementation in terms
of Definition \ref{def:F-simulation}.

In any case, these transformations define a partial order on channels, 
modulo equivalence. For each specific theory, the question arises
how to characterize this equivalence relation and the partial order.

\begin{exmp}
Consider the case of a resource theory with free channels
\[
  \mathfrak{F} := \{ \text{cptp maps preserving the maximally mixed state} \}.
\]
This is a subset of the resource theory of thermodynamics, to be precise
of the Gibbs-preserving maps, when the Hamiltonian is trivial.
For $A=B$, $\mathfrak{F}(A\longrightarrow A)$ consists precisely
of the unital cptp maps. For these, the transformability of 
states is well-known to be characterized precisely by the 
majorization relation \cite{marshall1974}: For $\rho,\sigma \in \cS(A)$
there exists a free channel $\cF$ such that $\cF(\rho)=\sigma$ iff $\rho \prec \sigma$.
What is the generalization of this criterion to channels
$\cM,\cN\in \text{CPTP}(A\longrightarrow A)$?
\end{exmp}

There is a simple case where the possibility of simulation could be directly related to the matrix representations of the unitary channels:
\begin{exmp}
Consider the resource theory of coherence, with free channels $\mathfrak{F} := \text{IO or SIO}$, and unitary channels $U, V$ acting on $k$-dimensional Hilbert space. 
A necessary condition for the simulation $U \stackrel{\mathfrak{F}}{\longrightarrow} V$ is given by majorization between the squared column vectors of the matrix representation of $U,V$ in the incoherent basis:
$\{|u_{i1}|^2,|u_{i2}|^2,...,|u_{ik}|^2\} \succ \{|v_{i1}|^2,|v_{i2}|^2,...,|v_{ik}|^2\}$ for all $i = 1,...,k$, where $u_{ij}, v_{ij}$ are matrix elements of $U,V$. 
This simply follows from that we must have $\Delta\bigl(U(\ketbra{j}{j})\bigr)\succ\Delta\bigl(V(\ketbra{j}{j})\bigr)$ for all incoherent basis states $\ket{j}$ \cite{PhysRevLett.116.120404}.
\end{exmp}


\section{Composition of multiple resources}
\label{sec:composition}
The issue of having access to several resources at the same time,
in the case of the resources being states, is satisfactorily
captured by the same tensor product that governs the composition
of systems: That is, to have access to states $\rho_i \in \cS(A_i)$ for
$i=1,\ldots,n$, can be equivalently described by having the
state $\sigma = \rho_1\ox\rho_2\ox\cdots\ox\rho_n \in \cS(A_1\ox\cdots\ox A_n)$.
The reason is that any free operation $\cF_i \in \mathfrak{F}(A_i\longrightarrow B_i)$
that could be used to transform $\rho_i$, is still available as
$\cF_i \ox \id_{A_{\neq i}} 
 \in \mathfrak{F}\left(A_{[n]} \longrightarrow B_i\ox A_{\neq i}\right)$,
where $A_{[n]} = \bigotimes_{j=1}^n A_j$ 
and $A_{\neq i} = \bigotimes_{j\in[n]\setminus i} A_j$, and so is
any tensor product, but of course potentially much more. Because of this,
a collection of states $\rho_i$ may without any loss of generality
be identified with the single tensor product state $\sigma$.

For channels, on the other hand, having access to quantum channels 
$\cN_1:A_1\longrightarrow B_1$ and $\cN_2:A_2\longrightarrow B_2$
as resources, could naturally mean that one can build $\cM = \cN_1\ox\cN_2$, but 
this does not capture the essence of having access to the two channels
separately, as in general $\cM$ cannot be used, in conjunction with
free resources, to build $\cN_2\circ\cF\circ\cN_1$, for a free
channel $\cF\in\mathfrak{F}(B_1\longrightarrow A_2)$, or similarly
$\cN_1\circ\cF'\circ\cN_2$, for a free channel $\cF'\in\mathfrak{F}(B_2\longrightarrow A_1)$,
or more complicated circuits built from free channels and taking $\cN_1$
and $\cN_2$ as black boxes in certain places. Conversely, neither
of these combinations is in general able to simulate the tensor product $\cM$,
nor each one the other. 
Note, however, that the tensor product at least mathematically
represents the pair $(\cN_1,\cN_2)$.
Indeed, the most general channel that can be created using $\cN_1$ and $\cN_2$
can be written as $\cF_2\circ(\cN_2\otimes\id)\circ\cF_1\circ(\cN_1\otimes\id)\circ\cF_0$,
or $\cF_2\circ(\cN_1\otimes\id)\circ\cF_1\circ(\cN_2\otimes\id)\circ\cF_0$.

\begin{figure}[htbp]
    \centering
    \includegraphics[scale=0.45]{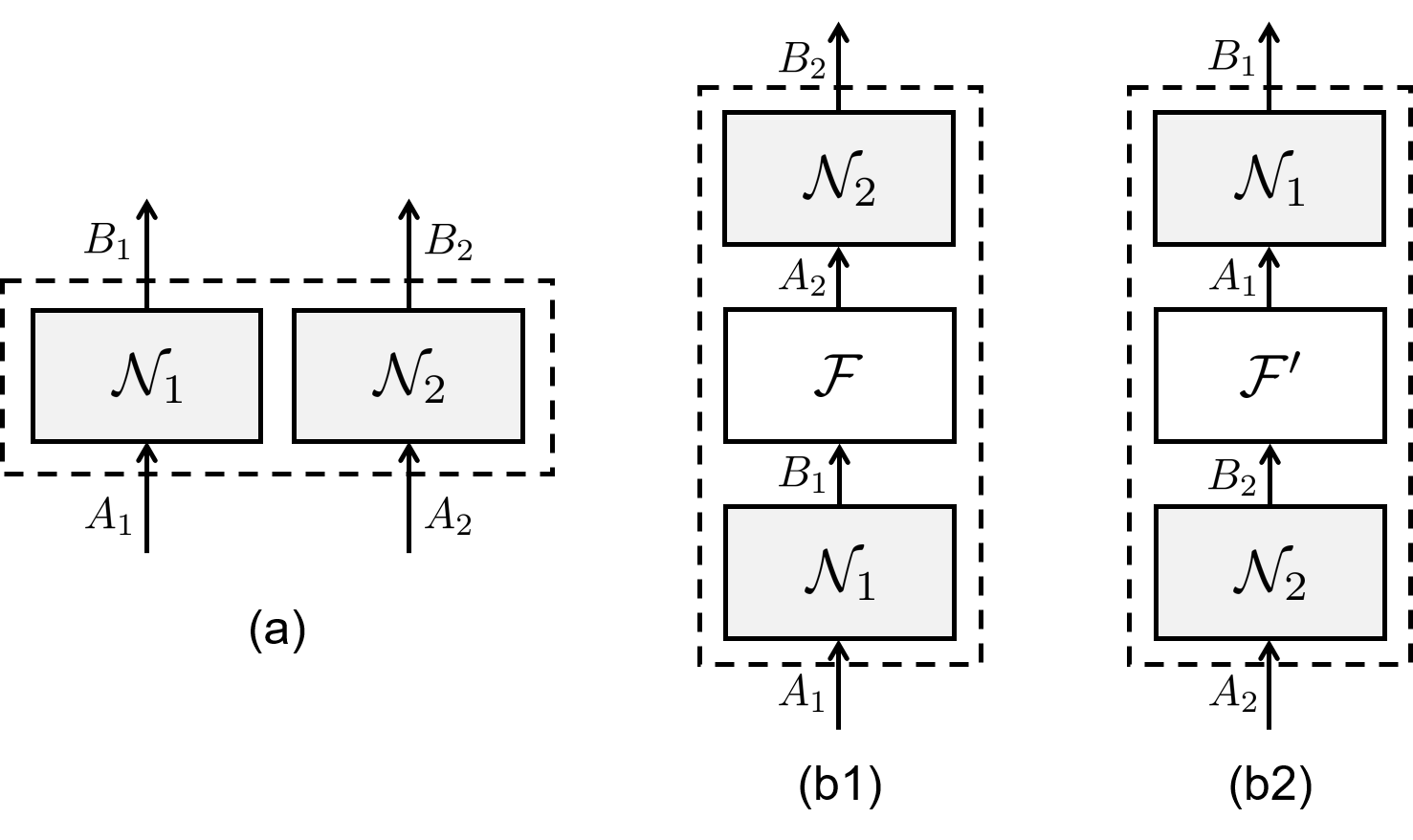}
    \caption{Different basic schemes of composing resource channels $\mathcal{N}_1,\mathcal{N}_2$. 
             (a) Parallel composition (tensor product): $\mathcal{N}_1\otimes\mathcal{N}_2$; 
             (b) Sequential composition, allowing intermediate free processings: 
             (b1)  $\cN_2\circ\cF\circ\cN_1$; 
             (b2) $\cN_1\circ\cF\circ\cN_2$. 
             In general, neither one can simulate another.}
    \label{fig:composition}
\end{figure}

That all this is not an esoteric theoretical issue, is demonstrated 
by the resource theories of entanglement \cite{bip},
quantum Shannon theory \cite{DHW:Shannon-resource}, quantum
thermodynamics \cite{thermal_c,FaistBertaTomamichel:thermo-cost},
and coherence theory \cite{coh_c,coh_mio}: in all these cases, channels
are used iteratively to distill, for instance, a unit resource
(entanglement, work, cosbits, etc). To be precise, while the most
general allowed protocol to generate bipartite entanglement from a given
bipartite unitary resource $U$ in \cite{bip} alternates one application
of the unitary with a general LOCC, a total of $n$ times, the
rate-optimal procedure partitions $n = n_1 + n_2 + \ldots + n_r$
into a fixed (or equivalently: arbitrarily slowly growing) number
of blocks, and in round $i\in[r]=\{1,\ldots,r\}$ acts with $U^{\ox n_i}$
on the state created in round $i-1$, and follows it with an LOCC
operation. While the blocks are not all equal in size, the rate
$R_i = \frac{n_i}{n}$ is a non-zero number in the asymptotics $n\rightarrow\infty$.
This pattern is found again in \cite{thermal_c}, where work is
extracted from a general channel $\cN$ by interleaving $\cN^{\ox n_i}$
with thermal operations; likewise in \cite{coh_c,coh_mio} in the case of
coherence generation. In the resource theory of quantum Shannon theory \cite{DHW:Shannon-resource}, 
this issue was somewhat avoided, by not allowing bipartite channels, but it could
not be kept out completely, since it showed up in the primitive of
``recycling'' resources. Thus, for asymptotic resources, 
$n \gg 1$ copies of $\cN$, it seems we may almost identify the
resource with the single channel $\cN^{\ox n}$, as long as we
restrict to protocol where the ``adaptiveness'', i.e.~the depth
$r$ of rounds in which we use the channels sequentially rather
than parallel, can be considered a constant (or arbitrarily slowly growing).

The latter is clearly not an option if we are in the regime of
small $n$ or when we are dealing with a set of completely general resource channels, 
either in the regime of small $n$ or asymptotically.
So, we have to consider from first
principles what it might mean to have access to channels
$\cN_1,\ldots,\cN_n$. To sharpen the concepts, it might help to 
think of the resource being provided by a ``server'', and a 
``client'' wishes to use it. The client may dispose already of some other
quantum circuitry, and wishes to call upon the server to provide 
$\cN_i$ at any given moment, so as to use it in any place and in
any order in their circuit. By restricting the kind of context
in which the client may use the individual resources, we can tune the
character of the joint object $\cN_1\ox\cdots\ox\cN_n$. We shall
distinguish three levels: the plain tensor product
$\cN_1\ox\cdots\ox\cN_n \in \text{CPTP}(A_{[n]}\longrightarrow B_{[n]})$,
which requires an input to all $n$ systems $A_i$ before being invoked;
the \emph{channel sequence} $(\cN_1,\ldots,\cN_n)$, which can be used
in the prescribed order; and the \emph{channel collection}
$\{\cN_i\}$, which may be used in any order. While the first
is naturally an element of the cptp maps from $A_{[n]}$ to $B_{[n]}$,
the second is naturally contained in the set of memory channels, which we denote
$\text{CPTP}_{m}(A_{\vec{[n]}}\longrightarrow B_{\vec{[n]}})$,
and the third is most naturally considered an element of the set
of no-signalling channels, denoted
$\text{CPTP}_{ns}(A_{{[n]}}\longrightarrow B_{{[n]}})$.

Each of the three classes comes with its own natural distinguishability
norm, which is defined via the largest class of contexts in which to
build a hypothesis test between the two objects. In the case of plain
cptp maps, $\text{CPTP}(A_{[n]}\longrightarrow B_{[n]})$, this is just 
the well-known diamond norm $\|\cdot\|_\diamond$.
In the case of memory channels, $\text{CPTP}_{m}(A_{\vec{[n]}}\longrightarrow B_{\vec{[n]}})$,
this is the largest trace norm difference between the outputs of
so-called quantum combs \cite{Pavia-combs,comb}, a memory channel itself, presented
as a circuit with $n$ placeholders in which the channels $A_i\rightarrow B_i$
have to be used in the prescribed order. We denote this $\|\cdot\|_{\diamond\rightarrow}$,
and like the diamond norm \cite{Watrous:cb_sdp} it is an SDP, because the relevant quantum combs have
an SDP characterization.
In the case of no-signalling channels, $\text{CPTP}_{ns}(A_{{[n]}}\longrightarrow B_{{[n]}})$,
there are actually two options: the first is operational, defining 
$\|X\|_{\diamond\leftrightarrow}$ has the maximum of $\|X^\pi\|_{\diamond\rightarrow}$
over all permutations $\pi\in S_n$ of the order of the $n$ systems; this is still an SDP,
but the additional maximization makes it problematic for large $n$. The second 
defines $\|X\|_{\diamond ns}$ as the maximum trace norm of the output of any 
generalised superchannel that maps no-signalling channels to no-signalling channels;
the latter is a strictly larger class as the former, as shown in recent work
by Salek \emph{et al.} \cite{Salek:superposition}. In any case, also this norm is an SDP.
Hence, for an element of the linear span of $\text{CPTP}_{ns}(A_{{[n]}}\longrightarrow B_{{[n]}})$, 
we get the chain of inequalities
\[
  \|X\|_{\diamond} \leq \|X\|_{\diamond\rightarrow} \leq \|X\|_{\diamond\leftrightarrow} \leq \|X\|_{\diamond ns}.
\]

In the following, we will only make use of the plain tensor product as a 
big quantum channel, and of the no-signalling set. The reason is, that
the ordered sequences, in the set of memory channels, do not compose 
in a commutative way. Indeed, for two cptp maps $\cN_1$ and $\cN_2$,
$(\cN_1,\cN_2)$ and $(\cN_2,\cN_1)$ are in general incomparable as channel
sequences. 

To define what it means that a channel collection simulates
(exactly or approximately) another channel collection, this is easy 
if we interpret the initial and final resources
$\cN_1\ox\cdots\ox\cN_n \in \text{CPTP}(A_{[n]}\longrightarrow B_{[n]})$
and $\cM_1\ox\cdots\ox\cM_m \in \text{CPTP}(A_{[m]}'\longrightarrow B_{[m]}')$,
respectively, as simple cptp maps, as we can apply the definition from the
previous section.
For the genuine collection case, where we interpret 
$\cN_1\ox\cdots\ox\cN_n \in \text{CPTP}_{ns}(A_{[n]}\longrightarrow B_{[n]})$
and $\cM_1\ox\cdots\ox\cM_m \in \text{CPTP}_{ns}(A_{[m]}'\longrightarrow B_{[m]}')$
as no-signalling channels, we have to be more careful. With the idea that
to ''have'' a collection of channels $(\cM_1,\ldots,\cM_m)$ means having
the ability to apply them in any context and any order, we say that
$\alpha=\cN_1\ox\cdots\ox\cN_n$ simulates $\beta=\cM_1\ox\cdots\ox\cM_m$ if for any
permutation $\pi\in S_m$, there is a permutation $\tau\in S_n$ and a quantum
comb $\mathcal{C}$ such that 
\[
  \beta^\pi := \cM_{\pi(1)}\ox\cdots\ox\cM_{\pi(m)} = \mathcal{C}(\alpha^\tau).
\]
In the case of approximate simulation, the relevant norm 
$\|X\|_{\diamond\leftrightarrow}$ or $\|X\|_{\diamond ns}$ is used.
Naturally, we are mostly interested in \emph{free} transformations $\mathfrak{F}$
in a given resource theory of channels; to capture that, we shall simply demand that
the combs $\mathcal{C}$ above can be built by quantum circuitry involving only
free channels. Note that by the structure theorem of \cite{Pavia-combs}, every
quantum comb has a circuit decomposition with feed-forward memory, so that all
we are asking in addition is that the circuit elements be free. This generalises
our previous definition of channel simulation in the case of a single cptp map.

\section{Resource measures and monotones}
\label{sec:measures}
Another theme of resource theories, somewhat dual to the transformation
question, is the quantification of resource, which is studied either axiomatically or operationally.  These two approaches aim at addressing the following basic questions respectively: What mathematical properties does a functional need to satisfy to qualify as a sensible resource measure, and how to construct such measures? 
How does the magnitude of a certain measure correspond to the value or performance of the resource in some operational task, i.e.~what is the operational interpretation of this measure?

In the well-established resource theories of states, these problems are quite 
thoroughly studied.  A basic property that a resource measure should exhibit is 
monotonicity, meaning that it should never increase under free operations. 
So resource measures are also commonly referred to as ``resource monotones''.  
For example, the minimum distance (given by some well-behaved metric) to the set 
of free states is an intuitive way to measure the degree of ``resourceness'' in general, 
and indeed typically a monotone. 
We also have considerable knowledge about the operational meanings of various resource measures, both in specific resource theories and in the general setting.  
A standard task is to identify a ``unit'' resource, or an otherwise desirable object, and study how efficiently a given resource can be converted into the unit, or vice versa how efficiently the unit resource can be used to create the given one, only by free operations.
In the one-shot settings, the resource value (determined by the optimal rates of accomplishing such conversion tasks) are typically given 
by some form of smooth min-/max-relative 
entropy of resource (distance to the set of free states) 
\cite{Datta:max,BrandaoDatta:one-shot,PhysRevLett.120.070403,2017arXiv171110512R,AnshuHsiehJain:erasure,1rt}; 
while in the asymptotic (many-copy, i.i.d.) settings, the (regularized) relative entropy of resource 
is found to play a distinguished and universal
role \cite{PhysRevLett.89.240403,HorodeckiOppenheim:resource-theories,bg}.

Here, we initiate similar studies for resource theories of quantum channels at a general level.
We first lay the faithfulness and monotonicity conditions that proper resource measures of channels should satisfy, and then propose two general types of resource monotones for quantum channels. In particular, 
we introduce the max-relative entropy between channels and the 
log-robustness of channels in a given resource theory (and their smoothed 
versions), which will find operational meanings in the next section.

\subsection{Conditions}
We now explicitly list the conditions for resource measures defined on quantum channels.
The main difference between channels and states is that, since now the objects 
are themselves channels, the measure is expected to be monotone under both 
free pre- and post-processing as well as tensoring with free channels, or even more generally under all free superchannels.
Let $\Omega$ be a functional mapping cptp maps to real numbers. To qualify as
a resource monotone, the following properties are deemed necessary in any case or sometimes desirable (labeled by `*'). 

\begin{enumerate}
    \item Normalization:  
         $\Omega(\cN) = 0$ 
          if $\cN\in\mathfrak{F}$, and $\Omega(\cN)\geq 0$ for all channels $\cN$.
          \smallskip\\
          \emph{A proper measure must vanish on free channels and always take non-negative values.}

    \item *Faithfulness:
          $\Omega(\cN) = 0$ if and only if $\cN\in\mathfrak{F}$, and $\Omega(\cN) > 0$ otherwise.
    
    \item Monotonicity under composition with any free $\cM\in\mathfrak{F}$:
          In detail, under
          \begin{enumerate}
         
           \item left composition:  $\Omega(\cM\circ \cN) \leq \Omega(\cN)$;
           \item right composition: $\Omega(\cN\circ \cM) \leq \Omega(\cN)$;
           \item tensoring: $\Omega(\cN\otimes\cM) \leq \Omega(\cN)$.
          \end{enumerate}
          \smallskip
          \emph{A proper measure should be monotone non-increasing under left and right 
          composition with free channels. as well as tensoring, respectively;
          this ensures that if $\cN'$ can be simulated by $\cN$, then 
          $\Omega(\cN') \leq \Omega(\cN)$.}
     
     \item *Convexity: 
           $\sum_i p_i \Omega(\cN_i) \geq \Omega\left(\sum_i p_i \cN_i\right)$, 
           where $\{p_i\}$ is a probability distribution over channels $\cN_i$.
           \smallskip\\ 
           \emph{One may want a good measure to be non-increasing under convex 
           combination/probabilistic mixing, since it is usually considered easy.}
\end{enumerate}

Some simple observations can be made straight away about the above conditions.  First, the nonnegativity property can follow from the nullity condition and monotonicity, and hence does not need to be independently enforced.  
Second, monotonicity under left-composition is an
equality if $\cM$ has a left inverse $\cM'\in\mathfrak{F}$, i.e.~$\cM'\circ\cM=\id$;
likewise, monotonicity under right-composition is an
equality if $\cM$ has a right inverse $\cM'\in\mathfrak{F}$, i.e.~$\cM\circ\cM'=\id$.
Third, monotonicity under tensoring is an equality if the resource theory has free states, and in fact is only needed with respect to $\id$ as the rest the left and right composition with free map
$\id\otimes\cM$.



\subsection{General measures}
Now we discuss several general constructions of resource measures 
that only depend on $\mathfrak{F}$ (not on free states or resource destroying maps), 
namely the generating-power-type measures, distance measures and 
robustness-type measures.  
The first two are peculiar as they depart from an already existing
monotone on states and extend its domain in a canonical way to channels.
The third is a purely geometric construction, which however are not built upon state measures.



\subsubsection{Generating power}
An intuitive way to lift resource measures of states to channels is to 
consider the maximum increase, or ``generation'', of the resource, as measured by 
some faithful resource monotone on states $\omega$ (monotone under $\mathfrak{F}$) 
induced by the action of the channel. 
 
More explicitly, for channel $\cN:A\longrightarrow B$, the \emph{increasing power} 
of $\cN$ (without auxiliary systems) and the complete version (with auxiliary systems), given by monotone $\omega$,  are defined as
\begin{align}
  \Omega_{\mathsf{ip},\omega}(\cN)   &:= \sup_{\rho} \big\{\omega\bigl(\cN(\rho)\bigr)-\omega(\rho)\big\},\\ 
  \Omega^*_{\mathsf{ip},\omega}(\cN) &:= \sup_{\rho} \big\{\omega\bigl(\cN\otimes\id(\rho)\bigr)-\omega(\rho)\big\},
\end{align}
where the optimization for $R$ is over $\rho$ on $A$, and the optimization for 
$\Omega^*$ is implicitly over all auxiliary systems $C$ and $\rho$ on $A\otimes C$.  
 
It is also sensible to consider the \emph{generating power} and its complete version, which measure the maximum amount of resource that can be created from free states (if they exist for the resource theory under study, otherwise such quantities are undefined). Now the optimization is over free states only:
\begin{align}
  \Omega_{\mathsf{gp},\omega}(\cN)   &:= \sup_{\rho \text{ s.t.}\, \omega(\rho)=0} \omega\bigl(\cN(\rho)\bigr),\\ 
  \Omega^*_{\mathsf{gp},\omega}(\cN) &:= \sup_{\rho \text{ s.t.}\, \omega(\rho)=0} \omega\bigl(\cN\otimes\id(\rho)\bigr).
\end{align}

In general, these quantities could be different, while we always have $\Omega_{\mathsf{gp},\omega}\leq \Omega_{\mathsf{ip},\omega}, \Omega^*_{\mathsf{gp},\omega}\leq \Omega^*_{\mathsf{ip},\omega}$ and $\Omega_{\mathsf{gp},\omega}\leq \Omega^*_{\mathsf{gp},\omega}, \Omega_{\mathsf{ip},\omega}\leq \Omega^*_{\mathsf{ip},\omega}$.     
In some cases the equality may hold.  For example, if $\omega$ is given by a distance measure which obeys triangle and data processing inequalities to the set of free states,  it holds that  $\Omega_{\mathsf{gp},\omega}= \Omega_{\mathsf{ip},\omega}$ \cite{channel_discrimination}.

Now we show that the complete increasing and generating powers are indeed proper 
measures of resource in quantum channels, satisfying the basic axioms.
\begin{thm}
If $\omega$ is a resource monotone on states, then
$\Omega_{\mathsf{ip},\omega}$, $\Omega^*_{\mathsf{ip},\omega}$, $\Omega_{\mathsf{gp},\omega}$, 
and $\Omega^*_{\mathsf{gp},\omega}$ are normalized, and monotone under left and 
right compositions with free channels.

The complete versions, $\Omega^*_{\mathsf{ip},\omega}$ and $\Omega^*_{\mathsf{gp},\omega}$,
are furthermore monotone (invariant if the resource theory has free states, and $\omega$ satisfies the reasonable condition that $\omega(\rho\otimes\sigma) = \omega(\rho)$ for any free state $\sigma$) under tensoring with free channels, and hence
can be called resource monotones for channels.
\end{thm}

\begin{proof}
We examine the properties one by one. 
\begin{enumerate}
    \item Normalization. 
      Since $\omega\bigl(\cN(\rho)\bigr) = \omega(\rho)$ must hold for fixed points of 
      $\cN$, which always exists, $\Omega_{\mathsf{ip},\omega}(\cN)\geq 0$.
      Suppose $\cN\in\mathfrak{F}$. Then by the monotonicity of $\omega$ under 
      $\mathfrak{F}$, we have $\omega\bigl(\cN(\rho)\bigr) - \omega(\rho) \leq 0$ 
      for all $\rho$.  So $\Omega_{\mathsf{ip},\omega}(\cN) =  0$.  
      Of course $\omega\bigl(\cN(\rho)\bigr)=0$ when already assuming $\omega(\rho)=0$, 
      so by definition, $\Omega_{\mathsf{gp},\omega}(\cN) =  0$.  
      The proof directly generalizes to the complete versions 
      $\Omega^*_{\mathsf{ip},\omega}(\cN),\Omega^*_{\mathsf{gp},\omega}(\cN)$ by 
      applying the above arguments to $\cN\otimes\id$ and invoking its freeness 
      when $\cN\in\mathfrak{F}$ (due to axioms 1 and 3 of free channels).

    \item Monotonicity. 
      Suppose $\cM\in\mathfrak{F}$.  
      \begin{enumerate}
        \item Left composition: 
          \begin{align}
            \Omega_{\mathsf{ip},\omega}(\cM\circ \cN) 
               &= \sup_\rho \big\{\omega\bigl(\cM\circ \cN(\rho)\bigr)-\omega(\rho)\big\} 
                = \omega\bigl(\cM\circ \cN(\tilde\rho)\bigr)-\omega(\tilde\rho) \\ 
               &\leq \omega\bigl(\cN(\tilde\rho)\bigr)-\omega(\tilde\rho)\\
               &\leq\sup_\sigma \big\{\omega\bigl(\cN(\sigma)\bigr)-\omega(\sigma)\big\}
               = \Omega_{\mathsf{ip},\omega}(\cN),
          \end{align}
          where in the first line $\tilde\rho$  attains the sup, and the second line 
          follows from the monotonicity of $\omega$.
          The proof for $\Omega_{\mathsf{gp},\omega}(\cM\circ \cN)\leq \Omega_{\mathsf{gp},\omega}(\cN)$ 
          directly follows by ignoring the negative term.
          The proof for $\Omega^*_{\mathsf{ip},\omega}, \Omega^*_{\mathsf{gp},\omega}$ 
          directly follows by applying the above arguments to the extended channels.

    \item Right composition:
      \begin{align}
        \Omega_{\mathsf{ip},\omega}(\cN\circ \cM) 
          &= \sup_\rho \big\{\omega\bigl(\cN\circ \cM(\rho)\bigl)-\omega(\rho)\big\} \\
          &\leq \sup_\rho \big\{\omega\bigl(\cN\circ \cM(\rho)\bigr)-\omega\bigl(\cM(\rho)\bigr)\big\}
                 + \sup_\sigma \big\{\omega\bigl(\cM(\sigma)\bigr)-\omega(\sigma)\big\}\\
          &=\sup_\rho \big\{\omega\bigl(\cN\circ \cM(\rho)\bigr)-\omega\bigl(\cM(\rho)\bigr)\big\}\\
          &\leq\sup_{\rho'} \big\{\omega\bigl(\cN(\rho')\bigr)-\omega(\rho')\big\}
           = \Omega_{\mathsf{ip},\omega}(\cN),
       \end{align}
       where the third line follows from that the latter sup in the second line is never 
       positive by monotonicity and attains zero by fixed points of $\cM$ or free states. 
       \begin{equation}
         \Omega_{\mathsf{gp},\omega}(\cN\circ \cM) 
           =    \sup_{\rho \text{ s.t.}\, \omega(\rho)=0} \omega\bigl(\cN\circ \cM(\rho)\bigr) 
           \leq \sup_{\sigma \text{ s.t.}\, \omega(\sigma)=0} \omega\bigl(\cN(\sigma)\bigr) 
           =    \Omega_{\mathsf{gp},\omega}(\cN),
       \end{equation}
       since $\omega\bigl(\cM(\rho)\bigr)=0$ given that $\omega(\rho)=0$.
       The proof for $\Omega^*_{\mathsf{ip},\omega}, \Omega^*_{\mathsf{gp},\omega}$ 
       directly follows by applying the above arguments to the extended channels.

    \item Tensor product:
      \begin{align}
        \Omega^*_{\mathsf{ip},\omega}(\cN\otimes \cM) 
           &= \Omega^*_{\mathsf{ip},\omega}\bigl((\cN\otimes \id)\circ(\id\otimes\cM)\bigr) \\
           &\leq \Omega^*_{\mathsf{ip},\omega}(\cN\otimes \id)\\
           &\leq \Omega^*_{\mathsf{ip},\omega}(\cN),
      \end{align}
      where the second line follows from $\id\otimes\cM\in\mathfrak{F}$ and monotonicity 
      under composition, and the third line follows from the definition of $\Omega^*_{\mathsf{ip},\omega}$.
      The same chain of inequalities directly applies to $\Omega^*_{\mathsf{gp},\omega}$ 
      (but not the non-complete $\Omega_{\mathsf{ip},\omega},\Omega_{\mathsf{gp},\omega}$, 
      since the last inequality does not necessarily hold).\\
      Under the assumptions that the resource theory has free states, and $\omega$ 
      satisfies $\omega(\rho\otimes\sigma) = \omega(\rho)$ for any free state $\sigma$, 
      we can further obtain equality since 
      \begin{align}
        \Omega^*_{\mathsf{ip},\omega}(\cN\otimes \cM) 
           &\geq  \omega\bigl(\cN\otimes\id_{\tilde C}(\tilde\rho)\otimes \cM(\sigma)\bigr) 
                  - \omega(\tilde\rho\otimes \sigma) \\
           &= \omega\bigl(\cN\otimes\id_{\tilde C}(\tilde\rho)\bigr) - \omega(\tilde\rho) 
            = \Omega^*_{\mathsf{ip},\omega}(\cN).
      \end{align}
      where the auxiliary space $\tilde C$ and $\tilde\rho\in\mathcal{S}(A\otimes \tilde C)$ 
      attain the sup in $\Omega^*_{\mathsf{ip},\omega}(\cN)$ and $\sigma$ is a free state. 
      The second line follows from the assumption on $\omega$.   
      Similarly,
      \begin{align}
        \Omega^*_{\mathsf{gp},\omega}(\cN\otimes \cM)
          \geq \omega\bigl(\cN\otimes\id_{\tilde C'}(\tilde\sigma)\otimes \cM(\sigma)\bigr)
          =    \omega\bigl(\cN\otimes\id_{\tilde C'}(\tilde\sigma)\bigr) 
          = \Omega^*_{\mathsf{gp},\omega}(\cN),
      \end{align}
      where the auxiliary space $\tilde C'$ and free state $\tilde\sigma\in\mathcal{S}(A\otimes \tilde C')$
      attain $\Omega^*_{\mathsf{gp},\omega}(\cN)$.
\end{enumerate}
\end{enumerate}
\end{proof}

Note that the increasing and generating powers are not faithful when $\mathfrak{F}$ 
is not maximal (in the sense discussed in Section \ref{sec:free}): 
they vanish for channels belonging to the extended maximal set of 
non-generating maps but not $\mathfrak{F}$.

Some of these quantities have already been studied for certain resources, such as bipartite entanglement \cite{bip}, coherence \cite{PhysRevA.92.032331,Garcia-Diaz:2016:NCP:3179439.3179441,BU20171670,coh_c,coh_mio}, and thermodynamic nonequilibrium \cite{nanothermo,PhysRevLett.111.250404,thermal_c,FaistRenner:thermo-cost,FaistBertaTomamichel:thermo-cost}.  
In particular, there are known operational interpretations for a few cases. For example:
\begin{itemize}
    \item In the resource theory of coherence, the capacity of generating maximally 
          coherent states by IO \cite{coh_c} and MIO \cite{coh_mio} is given by 
          $\Omega^*_{\mathsf{ip}}$ in terms of the relative entropy of coherence 
          (see Example \ref{exp:mio});
    \item In the resource theory of thermal nonequilibrium, the unique measure that 
          characterizes the reversible rate of work extraction and work cost of channels 
          is $\Omega^*_{\mathsf{ip}}=\Omega_{\mathsf{ip}}$
          in terms of the free energy, which is proportional to the relative entropy 
          distance from the Gibbs state \cite{FaistRenner:thermo-cost} 
          (see Example \ref{exp:thermo});
    \item In general, $\Omega_{\mathsf{ip}}=\Omega_{\mathsf{gp}}$ in terms of the trace norm 
          distance of resource characterizes the success probability of discriminating 
          resourceful channels from free ones \cite{channel_discrimination}.
\end{itemize}

Instead of maximizing over all inputs, it has also been proposed to consider the average over random inputs \cite{PhysRevA.62.030301,PhysRevA.95.052306}, but this will in general not yield a monotone.

\subsubsection{Distance measures}
Another strategy that can always be implemented is to consider an appropriate
distance measures, now defined on channels.
The minimum distance to $\mathfrak{F}$ intuitively captures in some sense
the ``resourceness'' of a channel.   
That is, we may consider 
\begin{equation}
  \Delta_\delta(\cN) := \inf_{\cL\in\mathfrak{F}} \delta(\cN,\cL),
\end{equation}
given by some distance measure between channels $\delta$, which has to
satisfy certain properties to yield a resource measure.

A fundamental example is the diamond norm, aka completely bounded trace norm:
\begin{equation}
    \Delta_\diamond(\cN) = \inf_{\cL\in\mathfrak{F}} \|\cN-\cL\|_\diamond 
              = \inf_{\cL\in\mathfrak{F}}\sup_\rho\|\cN\otimes\id(\rho) - \cL\otimes\id(\rho)\|_1.
\end{equation}
Note that here the extended space is important to relate the diamond norm to the task of channel discrimination, since entangled inputs may help \cite{watrous,piani-watrous}.

More generally, let us consider defining $\delta$ by lifting distance measures 
from states as follows:
\begin{align}
  \Delta_{(d)}(\cN)   &:= \inf_{\cL\in\mathfrak{F}}\sup_{\rho}d\bigl(\cN(\rho),\cL(\rho)\bigr),\\
  \Delta_{(d)}^*(\cN) &:= \inf_{\cL\in\mathfrak{F}}\sup_{\rho}d\Bigl(\cN\otimes\id(\rho),\cL\otimes\id(\rho)\Bigr),
\end{align}
where $d$ is some contractive distance measure on states (that is nonnegative and 
vanishes when the arguments are equal).  Here $d$ can be e.g.~the trace norm of the 
difference or a relative entropy.  

Such distance-type measures also satisfy the basic conditions for resource measures of channels: 
\begin{thm}\label{thm:distance}
If the distance measure $d$ obeys the data processing inequality, then
$\Delta_{(d)}$ and $\Delta_{(d)}^*$ are faithful, and monotone under left and right compositions with free channels.

Furthermore, $\Delta_{(d)}^*$ is also monotonic (invariant if the resource theory has free states) under tensoring with free channels.
\end{thm}

\begin{proof}
Again, we examine the properties one by one.
\begin{enumerate}
    \item Normalization and faithfulness. By definition, $\Delta_{(d)}(\cN),\Delta_{(d)}^*(\cN) = 0$ 
    iff $\cN\in\mathfrak{F}$ and positive otherwise.
    \item Monotonicity.   Suppose $\cM\in\mathfrak{F}$, and $\tilde{\cL}\in\mathfrak{F}$ is an optimal free channel that attains the inf in $\Delta_{(d)}(\cN)$. 
\begin{enumerate}
    \item Left composition: 
\begin{align}
\Delta_{(d)}(\cM\circ \cN)&=\inf_{\cL\in\mathfrak{F}}\sup_{\rho}d\bigl(\cM\circ \cN(\rho), \cL(\rho)\bigr)\\&\leq \sup_\rho d\bigl(\cM\circ \cN(\rho), \cM\circ \tilde{\cL}(\rho)\bigr) = d\bigl(\cM\circ \cN(\tilde\rho), \cM\circ \tilde{\cL}(\tilde\rho)\bigr) \\&\leq d\bigl(\cN(\tilde\rho), \tilde{\cL}(\tilde\rho)\bigr) \leq \sup_\sigma d\bigl(\cN(\sigma), \tilde{\cL}(\sigma)\bigr) = \Delta_{(d)}(\cN),
\end{align}
where the second line follows from $\cM\circ \tilde{\cL}\in\mathfrak{F}$ (due to axiom 1 of free channels), and the third line follows from the data processing inequality. 
The proof for $\Delta_{(d)}^*$ directly follows by applying the previous argument to the extended channels.

\item Right composition: 
\begin{align}
\Delta_{(d)}(\cN\circ \cM)&=\inf_{\cL\in\mathfrak{F}}\sup_{\rho}d\bigl(\cN\circ \cM(\rho), \cL(\rho)\bigr)\\&\leq \sup_\rho d\bigl(\cN\circ \cM(\rho), \tilde{\cL}\circ \cM(\rho)\bigr) 
\\&\leq\sup_\sigma d\bigl(\cN(\sigma), \tilde{\cL}(\sigma)\bigr) = \Delta_{(d)}(\cN),
\end{align}
where the second line follows from $\tilde{\cL}\circ \cM\in\mathfrak{F}$. 
The proof for $\Delta_{(d)}^*$ directly follows by applying the previous argument to the extended channels.

\item Tensor product:
\begin{align}
    \Delta_{(d)}^*(\cN\otimes\cM) &= \Delta_{(d)}^*\bigl((\cN\otimes \id)\circ(\id\otimes\cM)\bigr) \\&\leq \Delta_{(d)}^*(\cN\otimes \id)\\
    &\leq \Delta_{(d)}^*(\cN),
\end{align}
where the second line follows from $\id\otimes\cM\in\mathfrak{F}$ and monotonicity under composition, and the third line follows from the definition of $\Delta_{(d)}^*$.   Note that the last inequality does not necessarily hold for the non-complete version $\Delta_{(d)}$.  

We can further obtain equality when the resource theory has free states (the subscripts denote the input space of channels or the space of density operators):
\begin{align}
\Delta_{(d)}^*(\cN\otimes \cM) &= \inf_{\cL\in\mathfrak{F}}\sup_{\rho}d\bigl(\cN_A\otimes \cM_B\otimes\id_C(\rho_{ABC}), \cL_{AB}\otimes\id_C(\rho_{ABC})\bigr) \\
&\geq \inf_{\cL\in\mathfrak{F}}d\bigl(\cN_A\otimes\id_{\tilde{C}}(\tilde\rho_{A\tilde{C}})\otimes \cM(\sigma), \cL_{AB}\otimes\id_{\tilde{C}}(\tilde\rho_{A\tilde{C}}\otimes \sigma_B)\bigr)\\
&\geq \inf_{\cL'\in\mathfrak{F}}d\bigl(\cN_A\otimes\id_{\tilde{C}}(\tilde\rho_{A\tilde{C}}), \cL'_{A}\otimes\id_{\tilde{C}}(\tilde\rho_{A\tilde{C}})\bigr) = \Delta_{(d)}^*(\cN),
\end{align}
where the auxiliary space $\tilde C'$ and  $\tilde\rho$  attain $\Delta_{(d)}^*(\cN)$.  The second line follows from letting $\rho_{ABC} = \tilde\rho_{A\tilde{C}}\otimes \sigma_B$, and the third line follows from monotonicity of 
$d$ under partial trace.
\end{enumerate}
\end{enumerate}
\end{proof}


\subsubsection{Robustness and log-robustness}
Here we generalize another particularly important type of resource measure, 
namely robustness and log-robustness, to quantum channels.  
We first define the max-relative entropy between channels and the smoothed versions as follows:
\begin{defn}[Channel max-relative entropy]
In analogy to states, define the max-relative entropy between two cp maps 
$\cN$ and $\cM$ as
\begin{equation}
D_{\max}(\cN\|\cM)
   :=\log \min \{\lambda: \cN\leq \lambda \cM\},
\end{equation}
where the inequality sign refers to the complete-positivity order between
superoperators, meaning that the difference between r.h.s. and l.h.s. is
completely positive.
An equivalent form for channels (i.e.~cptp maps) is
\begin{equation}
  D_{\max}(\cN\|\cM) = -\log \max \{p\in[0,1]:\cM=p\cN+(1-p)\cN',\, \cN' \text{cptp}\}.
\end{equation}

The \emph{$\epsilon$-smooth max-relative entropy} is defined by minimizing 
within the $\epsilon$-ball given by the diamond norm ($0\leq\epsilon\leq 1$):
\begin{equation}
  D_{\max}^\epsilon(\cN\|\cM) 
     := \inf_{\cN': \frac{1}{2}\|\cN'-\cN\|_\diamond\leq\epsilon} D_{\max}(\cN'\|\cM).
\end{equation}
\end{defn}

Note that similar quantities have already appeared in \cite{coh_mio,2018arXiv180705354F}.      The relative entropies for channels are actually made in such a way as to generalize the well-known ones for states (such as the max-relative entropy $D_{\max}$ above and the relative entropy $D$ which will appear later), when states
can be viewed as cptp maps from a 1-dimensional Hilbert space.

We may also rewrite the max-relative entropy for channels in terms
of the max-relative entropy for states \cite{Datta:max}:
\begin{equation}\begin{split}
  D_{\max}(\cN\|\cM) 
      &= \log\min\left\{\lambda: \cN\otimes\id(\rho)\leq \lambda\bigl(\cM\otimes\id(\rho)\bigr), 
                                                                                             \forall\rho\right\} \\
      &= \sup_\rho \log\min\left\{\lambda:\cN\otimes\id(\rho)\leq \lambda\bigl(\cM\otimes\id(\rho)\bigr)\right\} \\
      &= \sup_\rho {D}_{\max}\bigl(\cN\otimes\id(\rho)\big\|\cM\otimes\id(\rho)\bigr).
  \label{state}
\end{split}\end{equation}
As a matter of fact, in the last line, we may choose $\rho$ to be
the maximally entangled state between the input system and an auxiliary
system of equal dimension, or indeed any pure entangled state of maximal
Schmidt rank. The reason is the Choi isomorphism, which translates 
the complete-positive order of quantum channels into semidefinite
order of matrices.   That is, $D_{\max}(\cN\|\cM)$ is simply the max-relative entropy between the Choi matrices of $\cN$ and $\cM$.

While we cannot obtain the analogous identity for the smooth max-relative entropy
of channels, we can get at least a lower bound:
\begin{equation}\begin{split}
  D_{\max}^\epsilon(\cN\|\cM) 
      &=    \inf_{\cN': \frac12\|\cN'-\cN\|_\diamond\leq\epsilon} D_{\max}(\cN'\|\cM) \\
      &=    \inf_{\cN': \frac12\|\cN'-\cN\|_\diamond\leq\epsilon} 
                \sup_\rho {D}_{\max}\bigl(\cN'\otimes\id(\rho)\|\cM\otimes\id(\rho)\bigr)         \\
      &\geq \sup_\rho \inf_{\cN': \frac12\|\cN'-\cN\|_\diamond\leq\epsilon} 
                          {D}_{\max}\bigl(\cN'\otimes\id(\rho)\|\cM\otimes\id(\rho)\bigr)         \\
      &\geq \sup_\rho \inf_{\sigma': \frac12\|\sigma'-\cN\otimes\id(\rho)\|_1\leq\epsilon} 
                          {D}_{\max}\bigl(\sigma'\|\cM\otimes\id(\rho)\bigr)                      \\
      &=    \sup_\rho {D}_{\max}^\epsilon\bigl(\cN\otimes\id(\rho)\|\cM\otimes\id(\rho)\bigr).
  \label{eq:state-smoothed}
\end{split}\end{equation}
Note that to obtain this inequality, we have to define the smooth
max-relative entropy with respect to an $\epsilon$-ball in trace
distance, and not, as is customary in one-shot quantum Shannon theory,
with respect to the purified distance \cite{Datta:max,Datta:max-ieee}.
Here, the third line is by the minimax inequality, and the fourth
by the definition of the diamond norm.

Now, the robustness of resource, which characterizes the smallest 
proportion of ``noise'' to be added to the resource to create a
free state, can be directly generalized to quantum channels. 
We define the robustness of channels in the following, as well as its 
variant, the log-robustness, which is equivalent to the minimum 
max-relative entropy with respect to free channels, i.e.~$\Delta_{D_{\max}}$.

\begin{defn}[Resource robustness and log-robustness of channels]
The (global) robustness of a channel $\cN$ is defined as
\begin{equation}
    R(\cN) := \min\left\{s\geq 0: 
                       \frac{1}{1+s}\cN + \frac{s}{1+s}\cN' \in \mathfrak{F},\ \cN' \text{ cptp} \right\}.
\end{equation}

The log-robustness of $\cN$ is defined as
\begin{align}
  LR(\cN) &:= \log(1+R(\cN)) \\
        &=  -\log \max \{ p\in[0,1] : p\cN+(1-p)\cN' \in \mathfrak{F},\, \cN' \text{ cptp} \} \\
        &=  \min_{\cM\in\mathfrak{F}} D_{\max}(\cN\|\cM),
\end{align}
where $R(\cN)$ is the robustness of $\cN$.

Furthermore, the $\epsilon$-smooth robustness and log-robustness are given by
\begin{align}
  R^\epsilon(\cN)  &= \inf_{\cN': \frac12\|\cN'-\cN\|_\diamond\leq\epsilon} R(\cN'), \\
  LR^\epsilon(\cN) &= \inf_{\cN': \frac12\|\cN'-\cN\|_\diamond\leq\epsilon} LR(\cN') \notag\\
                   &= \inf_{\cN': \frac12\|\cN'-\cN\|_\diamond\leq\epsilon} 
                                       \min_{\cM\in\mathfrak{F}} D_{\max}(\cN'\|\cM) \notag\\
                   &= \min_{\cM\in\mathfrak{F}} D_{\max}^\epsilon(\cN\|\cM).
\end{align}
\end{defn}
Both are convex optimization problems (except for the logarithm), if the free sets $\mathfrak{F}(A\longrightarrow B)$
are convex. Indeed, the objective function is linear, with constraints expressed 
by semidefinite conditions and membership in the convex cone
$\RR_{\geq 0}\mathfrak{F}(A\longrightarrow B)$.

It is straightforward to confirm that robustness and log-robustness of channels 
satisfy the necessary conditions for resource measures of channels:
\begin{prop}
\label{prop:robustness}
Both $R(\cN)$ and $LR(\cN)$ 
are faithful, and monotone under left and right compositions,
and invariant under tensor product with free channels if the resource
theory has free states.
\end{prop}
\begin{proof}
We have $LR(\cN) = \min_{M\in\mathfrak{F}} D_{\max}(\cN\|\cM) 
               = \Delta_{D_{\max}} = \Delta^*_{({D}_{\max})}$,
where the last step follows from Eq.~(\ref{state}).  
It is known that ${D}_{\max}(\rho\|\sigma)$ is nonnegative, 
vanishes when $\rho=\sigma$, and satisfies the data processing inequality 
\cite{Datta:max-ieee}. Therefore, by Theorem \ref{thm:distance}, 
$LR(\cN)=\Delta^*_{({D}_{\max})}$ satisfies the conditions.
\end{proof}

Remarkably, log-robustness is known as a particularly important characterization of 
general operational tasks in state resource theories:
\begin{itemize}
    \item The asymptotic transformation rate in the reversible framework allowing the maximal set of 
          free operations \cite{bg} (prefigured in the reversible theory of entanglement 
          \cite{BrandaoPlenio:ent-reversible,BrandaoPlenio:ent-reversible:long});
    \item The one-shot formation cost \cite{1rt};
    \item The catalytic resource erasure cost \cite{AnshuHsiehJain:erasure} (which we will 
          generalize to channels in the next section);
    \item The advantage of resource states in channel discrimination tasks \cite{2018arXiv180901672T}.
\end{itemize}
%
Furthermore, we know already that the smooth log-robustness for channels characterizes the one-shot simulation cost of channels in the following cases (see examples in Section \ref{sec:asymptotics} for more detailed discussions):
\begin{itemize}
    \item The resource theory of coherence with MIO as free channels \cite{coh_mio};
    \item  The resource theory of thermal nonequilibrium with Gibbs-preserving maps 
           as free channels \cite{FaistRenner:thermo-cost,FaistBertaTomamichel:thermo-cost},
           cf.~\cite{Faist:personal}.
\end{itemize}

In the next section, we shall establish a generic operational interpretation of channel 
log-robustness via the generalized task of resource erasure.

\section{One-shot resource erasure}
\label{sec:erasure}
Following Groisman \emph{et al.} \cite{GroismanPopescuWinter}, we consider
the question of how much randomness is required to turn a given resource
channel $\cN$ into a resource free channel $\cM \in \mathfrak{F}$. 
To formulate this as a well-defined problem, we adapt the framework 
of catalytic erasure, due to Anshu \emph{et al.} \cite{AnshuHsiehJain:erasure},
from states to channels.
The idea is to allow only free and at the same time reversible 
operations, or rather an ensemble of such operations, such that their
average turns $\cN$ into a free $\cM\in\mathfrak{F}$, or to an approximation 
of such an $\cM$. 
The randomness cost is then the logarithm of the number of elements in
the ensemble, or alternatively the entropy of the distribution.

Concretely, to make sense to the following definitions, we shall
assume that the resource theory at hand has the property that 
tensoring with a free channel is for free, that each system
has a free state and that the trace and partial trace is free
(the latter two mean that we can take marginals over parts of a 
system); furthermore, we shall require that in the tensor product
of many copies of the same system, the permutation unitaries are free. 

\begin{defn}[Catalytic resource destruction cost]
\label{defi:resource-erasure}
For a channel $\cN:A\longrightarrow B$ in a resource theory 
with free states $\mathfrak{F}$ satisfying the above assumptions,
we call an \emph{$\epsilon$-resource-destruction process} any free channel 
$\cF\in\mathfrak{F}(A'\longrightarrow B')$ together with an ensemble 
of pairs of free reversible channels $\{p_i,\cU_i^{AA'},\cV_i^{BB'}\}_{i=1}^k$
(i.e.~unitary conjugations in $\mathfrak{F}$ such that their inverses are 
in $\mathfrak{F}$, too), such that
\[
  \frac12\left\| \sum_{i=1}^k p_i \cV_i \circ (\cN\otimes\cF) \circ \cU_i - \cM \right\|_\diamond
     \leq \epsilon,
\]
for some free channel $\cM \in \mathfrak{F}(AA'\longrightarrow BB')$.

The minimum $\log k$ such that an $\epsilon$-resource-destruction
process exists, is denoted $\mathsf{COST}^\epsilon(\cN)$, and
called the $\epsilon$-resource-destruction cost.
\end{defn}

\medskip
The motivation for this definition stems from the idea that as a resource, 
$\cN$ is equivalent to $\cV_i \circ (\cN\otimes\cF) \circ \cU_i$.
Indeed, in the one direction, we can transform $\cN$ into $\cN\otimes\cF$
for free, and can compose it with the free maps $\cU_i$ and $\cV_i$.
Conversely, since the inverses of the latter two maps are assumed free, too, 
$\cV_i \circ (\cN\otimes\cF) \circ \cU_i$ can be transformed into $\cN\otimes\cF$ 
by free superchannels, and by composing with a free state on $A'$ 
and the partial trace over $B'$,
the latter is transformed to $\cN$. This means that the resource destruction
can be attributed entirely to the convex combination with probabilities
$p_i$, in other words the forgetting/erasure of $i$ \cite{GroismanPopescuWinter}.

\begin{rem}
We have used the diamond norm distance in the above definition,
to characterize the degree of approximation. As discussed earlier,
while this metric is clearly distinguished for single channel
resources, there are a range of possibilities when dealing
with multiple channels represented as a single multi-terminal
no-signalling correlation. We note that all of these, however, 
are base norms $\|\cdot\|_\bullet$ on the space of cp maps, or a linear subspace
thereof, and share the property that $\|\cN\|_\bullet = 1$ for
any cptp map or a tensor product of cptp maps. This will allow
us to bound these norms tightly under certain circumstances, even
if we are not always able to evaluate them concretely.

In particular, all of the norms discussed in Section \ref{sec:composition}
are eligible, and stronger than the diamond norm: recall that
$\|\cdot\|_{\diamond\leftrightarrow}$ and $\|\cdot\|_{\diamond ns}$
are upper bounds on $\|\cdot\|_{\diamond}$
\end{rem}

\begin{thm}
\label{thm:resource-erasure}
For any channel $\cN:A\longrightarrow B$ in a resource theory 
satisfying the assumptions 1, 2, 3, 4, 5, and 7, and any $0<\eta<\epsilon<1$,
\begin{equation}
  LR^{\mu\delta}(\cN) + \log\left(1-\frac{1}{\mu}\right)
         \leq \mathsf{COST}^\epsilon(\cN) 
         \leq LR^{\epsilon-\eta}(\cN)+2\log\frac{1}{\eta}-1,
\end{equation}
where $\delta = \sqrt{\epsilon(2-\epsilon)}$ and $\mu > 1$.
If the sets of free channels $\mathfrak{F}$ are convex, then the lower bound 
can be improved to $\mathsf{COST}^\epsilon(\cN) \geq LR^{\delta}(\cN)$.
\end{thm}

\begin{proof}
\emph{Upper bound (achievability):} 
Find a free channel $\cF_0 \in \mathfrak{F}(A\longrightarrow B)$ 
such that $LR^{\epsilon-\eta}(\alpha) = D_{\max}^{\epsilon-\eta}(\cN\|\cF_0) = D_{\max}(\cN'\|\cF_0)$,
i.e.~$\cF_0$ is the optimal free channel that achieves the minimum in the 
robustness and $\cN'$ is the optimal approximation of $\cN$:
$\frac12\|\cN-\cN'\|_\diamond \leq \epsilon-\eta$.
Consider the following resource destruction process: Let $A'=A^{\otimes n-1}$,
$B'=B^{\otimes n-1}$ and $\cF=\cF_0^{\otimes n-1}$, furthermore
$\cU_i$ the conjugation by the pair transposition $(1~i)$ 
between the first and the $i$th system in $AA' = A^{\otimes n}$,
and $\cV_i$ the conjugation by the pair transposition $(1~i)$ 
between the first and the $i$th system in $BB' = B^{\otimes n}$.
Thus,
\[
  \overline{\cN} = \frac1n \sum_{i=1}^n \cV_i \circ (\cN\otimes\cF) \circ \cU_i
                 = \frac1n \sum_{i=1}^n \cF^{\otimes i-1} \otimes \cN \otimes \cF^{\otimes n-i},
\]
has a distance from $\cM = \cF^{\otimes n}\in\mathfrak{F}(A^n\longrightarrow B^n)$ 
that is bounded by the generalized Convex-Split Lemma~\ref{gcs} (Appendix \ref{app:convex-split}). Namely, 
it guarantees that $\frac12\|\overline{\cN}-\cM\|_\diamond \leq \epsilon$ as
soon as $\log n \geq LR^{\epsilon-\eta}(\cN)+\log\frac{1}{4\eta^2}$. As a
matter of fact (cf.~the remark after Definition \ref{defi:resource-erasure}),
noting that both $\overline{\cN}$ and $\cM$ are $n$-fold resource objects,
and concretely are no-signalling correlations in $\text{CPTP}(A\longrightarrow B)^{\otimes n}$,
any valid base norm $\|\cdot\|_\bullet$ for their space is going to be
bounded by the same $\epsilon$, in particular 
$\|\cdot\|_{\diamond\leftrightarrow}$ and $\|\cdot\|_{\diamond ns}$.

\medskip
\emph{Lower bound (optimality):} 
We have to consider an arbitrary resource destruction process
consisting of a free catalyst $\cF \in\mathfrak{F}(A'\longrightarrow B')$
and an ensemble of free reversible unitary conjugations
$\{p_i,\cU_i^{AA'},\cV_i^{BB'}\}_{i=1}^k$, such that
for $\cN_i := \cV_i \circ (\cN\otimes\cF) \circ \cU_i$,
\[
  \frac12 \left\| \sum_{i=1}^k p_i \cN_i - \cM \right\|_\diamond \leq \epsilon,
\]
with a suitable free map $\cM\in\mathfrak{F}$.
By the well-known relations between trace norm and fidelity of
states \cite{Fuchs-vandeGraaf}, which extend to the diamond norm
and the completely bounded fidelity \cite{Kretschmann:Steinsprung,Shirokov:e-bounded-F}, 
this means
\[
  P_{\text{cb}}\left(\sum_{i=1}^k p_i \cN_i, \cM\right) \leq \delta = \sqrt{\epsilon(2-\epsilon)},
\]
with the completely bounded purified distance and the completely bounded fidelity defined as
\begin{align}
  P_{\text{cb}}(\cN,\cM) &= \sqrt{1-F_{\text{cb}}(\cN,\cM)^2}, \\
  F_{\text{cb}}(\cN,\cM) &= \inf_{\rho} F\bigl((\cN\otimes\id)\rho,(\cM\otimes\id)\rho\bigr),
\end{align}
respectively.

Choosing isometric dilations of the channels $\cN_i$ as
$W_i:AA' \hookrightarrow BB' \otimes E$, we can write down an
isometric dilation for $\overline{\cN}=\sum_i p_i\cN_i$, namely
$W := \sum_i \sqrt{p_i} W_i \otimes \ket{i}^F$, mapping $AA'$ to $BB'\otimes EF$.
By the Uhlmann theorem for the completely bounded fidelity \cite{Kretschmann:Steinsprung}
(see also the explanation in \cite[Prop.~1~\&{}~Appendix]{Shirokov:e-bounded-F}),
this means that we can find an isometric dilation $Z = \sum_i \sqrt{p_i} Z_i \otimes \ket{i}$ 
of $\cM$ that is $\delta$-close to $W$ w.r.t.~$P_{\text{cb}}$.
Tracing out $E$ and measuring $F$, and using the domination of the
trace distance by the completely bounded purified distance, we get
\[\begin{split}
  \delta &\geq \frac12 \left\| \sum_i p_i\cN_i \otimes \proj{i} 
                              - \sum_i p_i\cM_i \otimes \proj{i} \right\|_\diamond \\
         &=    \sum_i p_i \| \cN_i-\cM_i \|_\diamond                               \\
         &=    \sum_i p_i \| \cN\otimes\cF - \cV_i^{-1}\circ\cM_i\circ\cU_i^{-1} \|_\diamond \\
         &\geq \left\| \cN\otimes\cF - \sum_i p_i \cV_i^{-1}\circ\cM_i\circ\cU_i^{-1} \right\|_\diamond,
\end{split}\]
where $\cM_i=\tr_E W_i\cdot W_i^\dagger$ is a cp map (though not necessarily cptp). 

Now, we use two observations: 
First, we have $\sum_i p_i \cM_i = \cM \in \mathfrak{F}$,
hence $1+R(\cM_i) \leq \frac{1}{p_i}$. Sandwiching by the inverses of 
$\cU_i$ and $\cV_i$, respectively, this translates to
\[
  \cV_i^{-1}\circ\cM_i\circ\cU_i^{-1} \leq \frac{1}{p_i}\cV_i^{-1}\circ\cM\circ\cU_i^{-1}
\]
for all $i$. We first show how to get the simpler lower bound in the case
of convex $\mathfrak{F}$: Averaging over $i$ with weights $p_i$ yields
\[
  \sum_i p_i \cV_i^{-1}\circ\cM_i\circ\cU_i^{-1} \leq k\,\frac1k \sum_i \cV_i^{-1}\circ\cM\circ\cU_i^{-1},
\]
and by convexity of the set of free channels, the 
channel $\frac1k \sum_i \cV_i^{-1}\circ\cM\circ\cU_i^{-1}$ on the 
right hand side is free, hence we get
\[
  1+R^\delta(\cN\otimes\cF)
     \leq 1+R\left(\sum_i p_i \cV_i^{-1}\circ\cM_i\circ\cU_i^{-1}\right) \leq k.
\]
On the other hand, by our assumptions on the resource theory,
and Proposition \ref{prop:robustness}, we have
$R^\delta(\cN) = R^\delta(\cN\otimes\cF) = R^\delta(\cN_i)$ for all $i$,
so the claimed lower bound follows.

Moving on to the general case, let 
$\delta_i := \frac12 \| \cN\otimes\cF - \cV_i^{-1}\circ\cM_i\circ\cU_i^{-1} \|_\diamond$
in the above reasoning, so that we have $\sum_i p_i\delta_i \leq \delta$, which means
that the set $\Delta_\mu := \{i : \delta_i \leq \mu\delta\}$ has total
probability $>1-\frac{1}{\mu}$, for $\mu > 1$. At the same time,
the set $\Pi_\nu := \{i: \frac{1}{p_i} \leq \nu k\}$, for $\nu > 1$, has
total probability $>1-\frac{1}{\nu}$. 
Now, if $\frac{1}{\mu}+\frac{1}{\nu} \leq 1$, we know that there must be an
index $i \in\Delta_\mu\cap\Pi_\nu$, i.e. both $\delta_i \leq \mu\delta$
and $\frac{1}{p_i} \leq \nu k$. Thus, on the one hand,
\[
  1+R(\cM_i) = 1+R(\cV_i^{-1}\circ\cM_i\circ\cU_i^{-1}) \leq \nu k,
\]
and $\frac12 \| \cN\otimes\cF - \cV_i^{-1}\circ\cM_i\circ\cU_i^{-1} \|_\diamond \leq \mu\delta$,
allowing us to conclude that 
\[
  1+R^{\mu\delta}(\cN) \leq \nu k.
\]
Making the tightest possible choice, $\nu = \left(1-\frac{1}{\mu}\right)^{-1}$, we 
obtain the claim.
\end{proof}

\begin{rem}
To use the general lower bound, 
for instance for small $\delta$, we could think of $\mu$ as a constant
larger than $1$ and get the same form of the lower bound as in 
Theorem~\ref{thm:resource-erasure}, with a slightly larger smoothing error 
and an additive offset. On the other hand,
for $\delta$ close to $1$, one might choose $\mu = \frac{1+\delta}{2\delta}$, and get
smoothing error $\frac{1+\delta}{2}$ and additive offset $\log\frac{1-\delta}{1+\delta}$.
\end{rem}

\begin{rem}
We do not know, though suspect it to be the case, whether the converse (lower) bound of
Theorem~\ref{thm:resource-erasure} also holds in the case of multi-resources,
i.e. either a collection $\{\cN_1,\ldots,\cN_n\}$ of channels or a no-signalling
channel from $\text{CPTP}_{ns}(A_{[n]}\longrightarrow B_{[n]})$.
It seems that this would require a continuity bound of the Stinespring dilation of a 
memory channel.
\end{rem}

\section{To infinity -- and beyond!}
\label{sec:asymptotics}
In the last section, we considered the task of resource erasure
in the one-shot setting, i.e. one or finitely many copies of channels,
which we found to be tightly characterized by the smooth channel
log-robustness, which is derived from the channel max-relative
entropy.
By the knowledge of resource theories of quantum states, it is also natural to expect such quantities to 
play a fundamental role in the operational characterization of channel resources.
An important problem is then the asymptotic i.i.d. limit of the smooth log-robustness, 
which would be crucial to the operational resource theories of channels where 
asymptotically many instances are given (as conventionally assumed in information theory).

In the asymptotic regime of state resource theories, we typically see some relative 
entropy \cite{HorodeckiOppenheim:resource-theories,bg,AnshuHsiehJain:erasure}, 
which is not surprising due to the quantum asymptotic equipartition property (QAEP).  
Here we discuss the QAEP problem of channel max-relative entropy, which can be regarded as a simplified version, and perhaps an important stepping stone towards the problem of asymptotic log-robustness. Some partial results are presented, but the problem is not fully understood. 


Recall the following form of QAEP for the smooth max-relative entropy between states
(assuming $\mathrm{supp}\,\rho\subset\mathrm{supp}\,\sigma$) \cite{Datta:max-ieee}:
\begin{equation} 
  \lim_{n\rightarrow\infty} \frac{1}{n}{D}_{\max}^\epsilon(\rho^{\otimes n}\|\sigma^{\otimes n}) 
       = {D}(\rho\|\sigma),   \label{eq:qaep}
\end{equation}
for all $0<\epsilon<1$.
Now the goal is to study the same limit as on the left hand side for channels,
to be concrete
\begin{equation}
    \sup_{\epsilon> 0}\liminf_{n\rightarrow\infty}
        \frac{1}{n}D_{\max}^\epsilon(\cN^{\otimes n}\|\cM^{\otimes n}),
\end{equation}
and likewise for the $\limsup$. Can this be expressed in some closed form,
or at least in Shannon-theoretic terms as an optimization of
entropy or relative entropy expressions?

It is straightforward to obtain the following lower bound of nice form: 

\begin{thm}[Lower bound]
For all $0<\epsilon<1$, and any pair of channels $\cN,\cM:A\longrightarrow B$,
\begin{equation}\label{eq:lowerbound}
  \liminf_{n\rightarrow\infty}\frac{1}{n}D_{\max}^\epsilon(\cN^{\otimes n}\|\cM^{\otimes n}) 
              \geq  D(\cN\|\cM),
\end{equation}
where
\begin{equation}
   D(\cN\|\cM) :=           \sup_{\rho\in\cS(A\otimes C)} {D}\bigl(\cN\otimes\id_C(\rho)\big\|\cM\otimes\id_C(\rho)\bigr)
\end{equation}
is the quantum relative entropy of channels (previously defined in \cite{Cooney2016,PhysRevA.97.012332}).
\end{thm}
\begin{proof}
The proof follows from a simple reduction to the QAEP for states. For any test state $\rho\in\cS(A\otimes C)$,
\[\begin{split}
  \liminf_{n\rightarrow\infty}\frac{1}{n}D_{\max}^\epsilon(\cN^{\otimes n}\|\cM^{\otimes n})
    &= \liminf_{n\rightarrow\infty}\inf_{\cN':\frac{1}{2}\|\cN'-\cN^{\otimes n}\|_\diamond\leq\epsilon}
                                    \frac{1}{n}D_{\max}(\cN'\|\cM^{\otimes n})                         \\
    &= \liminf_{n\rightarrow\infty}\inf_{\cN':\frac{1}{2}\|\cN'-\cN^{\otimes n}\|_\diamond\leq\epsilon}
                                    \max_{\sigma\in\cS({A^nC})}
                   \frac{1}{n}{D}_{\max}\bigl(\cN'\otimes\id_C(\sigma)\big\|\cM^{\otimes n}\otimes\id_C(\sigma)\bigr) \\
    &\geq \liminf_{n\rightarrow\infty}\inf_{\cN':\frac{1}{2}\|\cN'-\cN^{\otimes n}\|_\diamond\leq\epsilon}
         \frac{1}{n}{D}_{\max}\left(\cN'\otimes\id_{C^n}(\rho^{\otimes n})\Big\|\bigl((\cM\otimes\id_C)\rho\bigr)^{\otimes n}\right) \\
    &\geq \liminf_{n\rightarrow\infty}
                           \inf_{\rho':\frac{1}{2}\|\rho'-((\cN\otimes\id_C)\rho)^{\otimes n}\|_1\leq\epsilon}
                   \frac{1}{n}{D}_{\max}\left(\rho'\Big\|\bigl((\cM\otimes\id_C)\rho\bigr)^{\otimes n}\right)                       \\
    &\geq 
    \lim_{n\rightarrow\infty}
         \frac{1}{n}{D}_{\max}^\epsilon\left(\bigl((\cN\otimes\id_C)\rho\bigr)^{\otimes n}\Big\|
                                      \bigl((\cM\otimes\id_C)\rho\bigr)^{\otimes n}\right) \\
    &=     
    {D}\bigl(\cN\otimes\id_C(\rho)\big\|\cM\otimes\id_C(\rho)\bigr),
\end{split}\]
where the second line follows from Eq.~(\ref{state}), the third line follows from letting $\sigma = \rho^{\otimes n}$, the fourth line follows from the max-min inequality, the fifth line follows from restricting $\rho_n$ to $\rho^{\otimes n}$, 
and the last line follows from Eq.~(\ref{eq:qaep}).
The conclusion Eq.~(\ref{eq:lowerbound}) is obtained by taking supremum over $\rho$.

\end{proof}



\bigskip\noindent

In the following, we discuss in more depth several informative specific theories, which yield 
special insights into the asymptotic theory of channel resources.  

First, the asymptotic theory exhibits simple and reversible behaviour for the elementary 
case of constant channels:
\begin{exmp}
Let $\mathfrak{F} := \{ \text{constant channels} \}$ 
(i.e. channels that map all input states to some given state).   
One can show that $LR(\cN) = \inf_{\cM\in\mathfrak{F}}D_{\max}(\cN\|\cM) = I_{\max}(A:B)_\cN$
(as well as the smoothed versions), where 
\begin{equation}
   I_{\max}(A:B)_\cN  :=  \inf_{\sigma_B}{D}_{\max}(\cN_{A'\longrightarrow B}(\Phi_{AA'})\|\tr_B\cN_{A'\longrightarrow B}(\Phi_{AA'})\otimes\sigma_B),
\end{equation}
$\Phi_{AA'}$ being the maximally entangled state on $AA'$,
is the channel's max-information of $\cN$ recently defined in \cite{2018arXiv180705354F}. 
Note that the channel smooth max-information $\frac{1}{2}I^\epsilon_{\max}(A:B)_\cN$ corresponds 
exactly to the $\epsilon$-error one-shot
simulation cost of $\cN$ under no-signalling assisted codes \cite{2018arXiv180705354F}.

The QAEP of $I^\epsilon_{\max}(A:B)_\cN$  (which is equivalent to a weaker quantum reverse 
Shannon theorem \cite{qrst,Berta2011} under free no-signalling correlations, see 
also \cite{2018arXiv180705354F} for a direct proof) then implies the ideal feature of 
the theory of constant channels $\mathfrak{F}$ that the log-robustness asymptotically 
converges to the relative entropy distance to the set of free channels:
\begin{equation}
    \lim_{\epsilon\rightarrow 0}\lim_{n\rightarrow\infty}\frac{1}{n}LR^\epsilon(\cN^{\otimes n}) 
                                               = \min_{\cM\in\mathfrak{F}}D(\cN\|\cM),
\end{equation}
which is also equivalent to the channel's mutual information.

This also gives to the r.h.s. an operational meaning as the asymptotic no-signalling-assisted quantum 
simulation cost, and the quantum reverse Shannon theorem implies equality of all 
no-signalling-assisted and entanglement-assisted quantum capacities and simulation costs.  
So these important quantum Shannon-theoretic quantities, and in particular 
the reversibility of the associated tasks, can also be understood as features 
of the resource theory of constant channels.
Note that the entanglement-assisted capacity theorem \cite{Bennett:2006:ECQ:2263245.2267047}
treats $\cN^{\otimes n}$ as a single cptp map, and that conversely the simulation
using entanglement \cite{qrst,Berta2011} and no-signalling correlation \cite{2018arXiv180705354F}
produces a single cptp map approximating $\cN^{\otimes n}$.
\end{exmp}

For state resource theories, it is well known that the maximal set of free operations 
(defined relative to convex set of free states, see Section \ref{sec:free}), modulo a
topological subtlety which seems necessary in complete generality, induces reversibility 
in the asymptotic limit \cite{HorodeckiOppenheim:resource-theories,BrandaoPlenio:ent-reversible,bg}.
However, it is not clear whether the analogous properties hold for the resource theory of channels.

In particular, we are interested in the following form of asymptotic reversibility for 
channels based on free channels (which is not quite the lifted notion based on a general set
of freeness-preserving superchannels): 
Given the definition of some reference resource or ``standard'' resource states that serves 
as a unit resource, does it hold that in the asymptotic limit, 
the minimum rate of unit resources needed to be consumed (as an attached ``battery'' 
system) to simulate any target channel by free channels (simulation cost $\zeta$) equals 
the maximum rate of unit resources that can be produced by this channel assisted by 
free channels (generating capacity $\xi$)?.
We say the theory is asymptotically reversible if $\zeta = \xi = \mu$, where
for any asymptotic i.i.d~multi-resource transformation 
$\cN^{\otimes n}\stackrel{\mathfrak{F}}{\longrightarrow} 
                 \stackrel{\epsilon}{\approx}\cN'^{\otimes nR}$
as $n\rightarrow\infty$ and $\epsilon\rightarrow 0$, can be achieved at rate
\begin{equation}
    R \rightarrow R(\cN\longrightarrow\cN') = \frac{\mu(\cN)}{\mu(\cN')},
\end{equation}
by first generating $\mu(\cN)$ units of standard resource per copy by $\cN$, and then 
simulating $\frac{n\mu(\cN)}{\mu(\cN')}$ copies of $\cN'$ using $\mu(\cN')$ units per copy.
Therefore, $R(\cN\longrightarrow\cN')R(\cN'\longrightarrow\cN) = 1$,
and $\mu(\cN)$ is the unique measure that characterizes the resource value of channel 
$\cN$ in terms of channel simulation.
In all known instances, $\cN^{\otimes n}$ is actually not considered as the most
general multi-resource, which would allow any $n$ uses of $\cN$, but as the more
modest multi-resource that is a sequence of $r$ cptp maps $\cN^{\otimes n}$, 
$n_1+n_2+\ldots+n_r=n$ and a fixed or arbitrarily slowly growing depth
$r$. This notion is formalized for instance in \cite{DHW:Shannon-resource}.

\medskip
When does this occur? In particular, does reversibility already hold when the set 
of free channels is maximal?
Here we elaborate on this issue in the physically motivated theories of 
quantum thermodynamics and coherence, which arise naturally from their respective 
sets of free states.

\begin{exmp}[\cite{FaistBertaTomamichel:thermo-cost}]\label{exp:thermo}
Maximal resource theory of quantum thermodynamics:
\[
  \mathfrak{F} := \mathfrak{G} 
                = \{ \text{Gibbs-preserving maps} \}
\]
(i.e. channels with fixed point the Gibbs state $\tau_\beta = \frac1Z e^{-\beta H}$, 
given inverse temperature $\beta$ and Hamiltonian $H$). 
The natural reference resource here is the ``quantum work'' needed to be expended
(or can be extracted) to implement (or from) a channel 
\cite{thermal_c,FaistRenner:thermo-cost,FaistBertaTomamichel:thermo-cost}.
The standard resource states storing the work are w.l.o.g.~a given set of energy levels $\{\ket{E}\}$.

This theory is indeed reversible: the asymptotic generating (work extraction) capacity and simulation 
cost are universally given by the free energy increasing power, aka the \emph{thermodynamic capacity}
\begin{equation}
    \mu(\cN) = \Omega_{\mathsf{ip},F}(\cN) = \max_{\rho}\left\{ F\bigl(\cN(\rho)\bigr) - F(\rho)  \right\},  \label{eq:th_cap}
\end{equation}
where $F(\rho) = \tr(\rho H) - \beta^{-1}S(\rho) = \beta^{-1}{D}(\rho\|\tau_\beta)$ is the free energy.
In \cite{thermal_c} it is shown how to extract $\mu(\cN)$ units of energy per use of the 
channel $\cN$, via a protocol of $r$ iterative invocations of the cptp map
$\cN^{\otimes \frac{n}{r}}$, when $r$ is sufficiently large and $n\rightarrow\infty$. No 
arbitrarily adaptive protocol can be better.
Conversely, in \cite{FaistBertaTomamichel:thermo-cost} the simulation cost is obtained by an 
AEP of the one-shot simulation cost, which actually reduces to the smooth log-robustness \cite{Faist:personal}.  
That is, the asymptotic log-robustness in this theory is given by Eq.~(\ref{eq:th_cap}),
and hence the simulation of cost of any $r$ multi-resource copies of $\cN^{\otimes \frac{n}{r}}$,
as $n\rightarrow\infty$.
\end{exmp}

\begin{exmp}\label{exp:mio}
Maximal resource theory of coherence:
$\mathfrak{F} := \text{MIO} = \{ \text{maximally incoherent operations} \}$.
The following results and a similar discussion is are from the recent paper \cite{coh_mio}.
The standard resource state in dimension $k$ is the maximally coherent state 
$\ket{\Psi_k} = \frac{1}{\sqrt{k}}\sum_{i=0}^{k-1}\ket{i}$, which represent $\log k$ ``cosbits''.  

Now consider the two directions separately:
\begin{enumerate}
\item Generating capacity. The proof for IO in \cite{coh_c} can be fairly easily extended 
      to MIO to show \cite{coh_mio}
      \begin{equation}
        \xi(\cN) = \Omega^*_{\mathsf{ip},C_r}(\cN) 
                 = \sup_{\rho}\left\{C_r\bigl(\cN\otimes \id(\rho)\bigr) - C_r(\rho)\right\},
      \end{equation}
      where $C_r(\rho) = \min_{\delta\in\Delta}{D}(\rho\|\delta)  
                       = {D}\bigl(\rho\|\Delta(\rho)\bigr) = S\bigl(\Delta(\rho)\bigr)-S(\rho)$ 
      is the relative entropy of coherence.
      Note that, as in the previous example, $\xi(\cN)$ units of coherence (``cosbits'')
      per channel use can be extracted, via a protocol of $r$ iterative invocations of 
      the cptp map $\cN^{\otimes \frac{n}{r}}$, when $r$ is sufficiently large and 
      $n\rightarrow\infty$. No arbitrarily adaptive protocol can be better.

\item Simulation cost. For this direction, we are also able to show that the one-shot 
      approximate MIO simulation cost $\zeta_1^\epsilon(\cN)$ of a channel $\cN:A\longrightarrow B$
      is given by its log-robustness:
      \begin{equation}
        LR^\epsilon(\cN) \leq \zeta_1^\epsilon(\cN) 
                         \leq \log\lceil 1+R^\epsilon(\cN)\rceil.
      \end{equation}
      By this result, the asymptotic simulation cost is given by the asymptotic log-robustness:
      \begin{equation}
         \zeta(\cN) = \sup_{\epsilon> 0}\limsup_{n\rightarrow\infty}\frac{1}{n} LR^\epsilon(\cN^{\otimes n}).
      \end{equation}
      Hence, as in the previous example, this is the simulation of cost of any $r$ multi-resource copies 
      of $\cN^{\otimes \frac{n}{r}}$, as $n\rightarrow\infty$.
      For the special case of cq-channels $\cN(\ket{i}\bra{j}) = \delta_{ij}\sigma_i$ 
      for some fixed $\sigma_i$ (written in the incoherent basis), 
      one can show that $\zeta(\cN)=\max_i C_r\bigl(\cN(\ket{i}\bra{i})\bigr)$.  
      But the problem of what this limit converges to in general remains open.  
      We wonder whether the result for Gibbs-preserving maps (Example \ref{exp:thermo}) can be indicative.
\end{enumerate}

It is not clear at this point whether this MIO theory of coherence is reversible or not in general.  
Either case seems very interesting: it could be that the reversibility holds, 
i.e.~$\xi(\cN) = \zeta(\cN)$ for all channels $\cN$ after all, which seems highly 
nontrivial to show and may lead to important advances in the understanding of channel 
theories (a partial result is that $\xi(\cN) = \zeta(\cN)=\max_i C_r\bigl(\cN(\ket{i}\bra{i})\bigr)$ 
for cq-channels $\cN$); or that the reversibility fails (even when the set of free 
channels is maximal), i.e.~$\xi(\cN) < \zeta(\cN)$ for some $\cN$, which would be 
a peculiar feature of the channel theory.   
\end{exmp}

\section{Conclusions}
In this work, we have drawn the outlines of a general framework of resource theories
where both the objects and their transforming free entities are quantum
channels. Such resource theories have been studied in concrete instances before,
but here we have attempted to identify the overarching common features of those.

In particular, we have exhibited the minimal requirements for a set of free
channels, and how to formulate resource transformations in a framework of
channel simulation. Furthermore, we have discussed generally available classes
of resource monotones. Most importantly among those, every quantum channel
resource theory has a notion of robustness, and we have shown that this number
can be given an operational interpretation as the resource erasure cost under
free operations and catalysts, counting only the randomness.

Much work remains to be done, especially in the domain of asymptotic transformations
and their rates. A most interesting problem is that of asymptotic reversibility,
which has been proven in some special cases, but remains wide open in general.

\bigskip
\emph{Note added.} 
As we were finalising the present manuscript, we became aware of the closely related 
work ``Operational Resource Theory of Quantum Channels'' by Yunchao Liu and Xiao Yuan, 
arXiv:1904.02680.

\begin{acknowledgments}
We thank Nilanjana Datta, Philippe Faist, Kun Fang, Aram Harrow, Seth Lloyd, Iman Marvian, Peter Shor, Xin Wang, 
and Dong Yang, for various discussions. In particular we are indebted to Gilad Gour for 
his illuminating insights into this topic, which has helped us infinitely in
focusing our own ideas.
ZWL acknowledges support by AFOSR, ARO, and Perimeter Institute for Theoretical Physics.
Research at Perimeter Institute is supported by the Government
of Canada through Industry Canada and by the Province of Ontario through the Ministry
of Research and Innovation.
AW acknowledges support by the Spanish MINECO (project FIS2016-80681-P),
with the support of FEDER funds, and the Generalitat de
Catalunya (CIRIT project 2017-SGR-1127).
\end{acknowledgments}

\appendix
\section{Conic convex-split lemma}
\label{app:convex-split}
Here we state and prove the abstract convex-split lemma, which in the case of states
is due to Anshu et al. \cite{AnshuDevabathiniJain:convex-split}, in a setting
where the objects are elements of an ordered real vector space $V$, which
is given by a positive cone $V_+ \subset V$, 
a closed, convex, generating and pointed cone in $V$. The relation
$\alpha\leq\beta$ is read to mean $\beta-\alpha\in V_+$, and it defines
vector partial order on $V$.
To make a meaningful statement, we ned two further structures:
first, an extension of the order from $V$ to $V^{\otimes n}$,
and second a norm $\|\cdot\|_\bullet$ norm on $V$ and all the $V^{\otimes n}$,
which ``goes well'' with the vector ordering, in the following sense:
\begin{enumerate}
  \item $V_+^{\otimes n} \subset (V^{\otimes n})_+$, i.e.~for all
    $\alpha_i$ positive in $V$, the tensor product 
    $\alpha_1\ox\alpha_2\ox\cdots\ox\alpha_n$ is positive in $V^{\ox n}$
  \item For all $\alpha_i\in V_+$, 
    $\|\alpha_1\ox\alpha_2\ox\cdots\ox\alpha_n\|_\bullet \leq \prod_{i=1}^n \|\alpha_i\|_\bullet$.
  \item More generally, if $\alpha_i\in V$ and $\alpha_j\in V_+$ for all $j\neq i$, 
    $\|\alpha_1\ox\alpha_2\ox\cdots\ox\alpha_n\|_\bullet \leq \prod_{i=1}^n \|\alpha_i\|_\bullet$.
\end{enumerate}

We call the pair $(V_+,\|\cdot\|_\bullet)$, which is really meant to include 
the tower $((V^{\otimes n})_+,\|\cdot\|_\bullet)$ of tensor extensions,
for all $n$, \emph{elegant} if these conditions are fulfilled. 

Examples of elegant ordered normed spaces include the Hermitian
trace class operators with the semidefinite order and trace norm,
and more generally any generalized probabilistic theory (GPT), which
is characterized by a positive cone $V_+ \subset V$ and a unit
functional $u\in V^*$ that supports $V_+$, 
$(V^{\otimes n})_+ \supset V_+^{\otimes n}$ with unit $u^{\otimes n}$,
and the \emph{base norm} 
\[
  \|x\|_\bullet = \inf u(\alpha+\beta) \text{ s.t. } x=\alpha-\beta,\ \alpha,\beta \in V_+.
\]
Most important for us is the example of $V$ being the set of Hermitian-preserving
superoperators, $V_+$ the cone of completely positive maps,
and $\|\cdot\|_\bullet = \|\cdot\|_\diamond$ the diamond norm. 
Furthermore, $V$ could be the subspace of Hermitian-preserving,
trace-proportional superoperators, i.e.~$\tr\cN(\rho) = \lambda_{\cN}\tr\rho$
for all $\rho$, with a constant $\lambda_{\cN}$.
Then, $V^{\ox n}$ consists of the no-signalling channels
from $n$ systems $A_1,A_2,\ldots,A_n$ to $n$ systems $B_1,B_2,\ldots,B_n$,
and $\|\cdot\|_\bullet$ could be the distinguishability norm on these
channels, either $\|\cdot\|_{\diamond\leftrightarrow}$ or $\|\cdot\|_{\diamond ns}$
which is defined via a maximization over all admissible
measurements on these objects, including all so-called wirings of the
various registers in a general quantum circuit, followed by a POVM
(see Section \ref{sec:transformations}). 

\begin{lem}[Generalized Convex-Split Lemma]
\label{gcs}
Let $\alpha,\beta \in V_+ \subset V$ with $\|\alpha\|_\bullet = \|\beta\|_\bullet = 1$,
and such that there exists an $\alpha'\in V_+$, $\|\alpha'\|_\bullet \leq 1$,
with $\beta = \frac{1}{\lambda}\alpha+\left(1-\frac{1}{\lambda}\right)\alpha'$.
Define
\begin{equation}
  \gamma^{(n)} = \frac{1}{n}\sum_{i=1}^{n}\alpha_i \otimes \beta^{\otimes[n]\setminus i},
\end{equation}
the average state of $\alpha$ in a random register $i\in_R[n]$ and
$\beta$ in all other registers $[n]\setminus i$. Then,
\begin{equation}
    \| \gamma^{(n)} - \beta^{\otimes n}\|_\bullet \leq \sqrt{\frac{\lambda}{n}}. 
\end{equation}
In particular, if $\log n \geq \log \lambda + 2\log\frac{1}{\delta}$, then
\begin{equation}
  \| \gamma^{(n)} - \beta^{\otimes n}\|_\bullet \leq \delta. 
\end{equation}
\end{lem}
\begin{proof}
The assumption can be rewritten as $\beta = p\alpha + (1-p)\alpha'$ for $\alpha'\in V_+$, and 
$p = \frac{1}{\lambda}\in(0,1]$; the case $p=0$ need not concern us.
Now, expand $\beta^{\otimes n}$ and $\gamma^{(n)}$ in terms of $\alpha$ and $\alpha'$:
\begin{align}
    \beta^{\otimes n} &= \bigl(p\alpha+(1-p)\alpha'\bigr)^{\otimes n} \\
    &= \sum_{S\subset[n]} p^{|S|}(1-p)^{n-|S|}\alpha^{\otimes S}\alpha'^{\otimes[n]\backslash S} \\
    &= \sum_{k=0}^n\binom{n}{k}p^k(1-p)^{n-k}\omega_k,
\end{align}
and similarly,
\begin{align}
    \gamma^{(n)} &= \frac{1}{n}\sum_{i=1}^n\alpha_i\otimes\bigl(p\alpha+(1-p)\alpha'\bigr)^{\otimes [n]\setminus i}\\
    &=\sum_{\emptyset\neq S\subset[n]} p^{|S|-1}(1-p)^{n-|S|}\alpha^{\otimes S}\alpha'^{\otimes[n]\backslash S}
    \\ &=\sum_{k=1}^n\binom{n-1}{k-1}p^{k-1}(1-p)^{n-k}\omega_k
               \\ &= \sum_{k=0}^n\frac{k}{np}\binom{n}{k}p^k(1-p)^{n-k}\omega_k,
\end{align}
where $\omega_k$ is the type class state obtained by putting all $S$ with
fixed cardinality $|S|=k$ together:
\begin{equation}
    \omega_k = \frac{1}{\binom{n}{k}}\sum_{|S|=k}\alpha^{\otimes S} \otimes \alpha'^{\otimes[n]\setminus S},
\end{equation}
and we used the identity $\binom{n-1}{k-1} = \frac{k}{n}\binom{n}{k}$.

Hence, we get
\begin{equation}
    \gamma^{(n)} - \beta^{\otimes n} 
             = \frac{1}{p}\sum_{k=0}^n\binom{n}{k}p^k(1-p)^{n-k}\left(\frac{k}{n}-p\right)\omega_k.
\end{equation}
Applying triangle inequalities on the norm, and using $\|\omega_k\|_\bullet \leq 1$,
we get
\begin{equation}
    \|\gamma^{(n)} - \beta^{\otimes n}\|_\bullet 
             \leq \frac{1}{p} \sum_{k=0}^n\binom{n}{k}p^k(1-p)^{n-k}\left|\frac{k}{n}-p\right|.
\end{equation}
Notice that the right hand side is a simple probabilistic object, namely
$\frac{1}{p}\mathbb{E}\left|\frac{K}{n}-p\right|$, where $K$ is a binomially distributed random variable: $K = \sum_{i=1}^n K_i$ with
i.i.d.~Bernoulli variables $K_i\sim B(1,p)$, each of which has mean $p$ and variance $p(1-p)$.  
So by convexity and standard facts from probability theory (Chebyshev inequality),
\begin{align}
    \|\gamma^{(n)} - \beta^{\otimes n}\|_\bullet 
        &\leq \frac{1}{p}\sqrt{\mathbb{E}\left(\frac{k}{n}-p\right)^2} \\
        &=    \frac{1}{p}\sqrt{\frac{1}{n^2}n\mathrm{Var}K_i} \\
        &=    \frac{1}{p}\sqrt{\frac{1}{n}p(1-p)} 
        \leq \sqrt{\frac{1}{pn}} = \sqrt{\frac{\lambda}{n}},
\end{align}
and we are done.
\end{proof}

\begin{rem}
As pointed out before, the above generalizes the Convex-Split Lemma 
of \cite{AnshuDevabathiniJain:convex-split}, when we take for $\|\cdot\|_\bullet$
the trace norm on trace class matrices.
Our applications are based on the choices $\|\cdot\|_\diamond$ and its
generalizations for multi-resources and no-signalling channels.
\end{rem}


\end{document}